\documentclass[preprint,12pt]{elsarticle}

\usepackage{amssymb}
\usepackage{amsmath,amsfonts}
\usepackage{amsthm}

\usepackage{algorithm}
\usepackage{algorithmicx}
\usepackage{algpseudocode}
\usepackage{setspace}
\usepackage{bm}
\usepackage{pifont}
\usepackage{mathrsfs}

\usepackage{hyperref}

\usepackage{threeparttable}
\usepackage{booktabs}
\usepackage{float}
\usepackage{multirow}
\usepackage{longtable}
\usepackage{array}
\usepackage{listings}
\usepackage{url}
\usepackage{xcolor}
\usepackage{textcomp}

\usepackage[utf8]{inputenc}
\usepackage[final,tracking=true,kerning=true,spacing=true]{microtype}

\graphicspath{{Fig/}} 
\usepackage{graphicx}

\usepackage{subcaption} 

\newtheorem{theorem}{Theorem}
\newtheorem{proposition}[theorem]{Proposition}

\newcommand{\hyem}{\text{HyEm}} 
\newcommand{\tspace}{\mathcal{T}}
\newcommand{\Hspace}{\mathbb{H}}
\newcommand{\R}{\mathbb{R}}
\newcommand{\Ball}{\mathbb{B}}

\journal{ScienceDirect}

\begin{document}

\begin{frontmatter}

\title{HyEm: Query-Adaptive Hyperbolic Retrieval for Biomedical Ontologies via Euclidean Vector Indexing}

\author[label1]{Ou Deng\corref{cor1}}
\ead{dengou@toki.waseda.jp}

\affiliation[label1]{organization={Graduate School of Human Sciences, Waseda University},
            addressline={2-579-15 Mikajima}, 
            city={Tokorozawa},
            postcode={359-1192}, 
            state={Saitama},
            country={Japan}}

\author[label2]{Shoji Nishimura}
\author[label2]{Atsushi Ogihara}
\author[label2]{Qun Jin\corref{cor1}}
\ead{jin@waseda.jp}

\affiliation[label2]{organization={Faculty of Human Sciences, Waseda University},
            addressline={2-579-15 Mikajima}, 
            city={Tokorozawa},
            postcode={359-1192}, 
            state={Saitama},
            country={Japan}}

\cortext[cor1]{Corresponding authors}

\begin{abstract}
Retrieval-augmented generation (RAG) is increasingly used to ground large language models (LLMs) in biomedical knowledge.
A recurring challenge is hierarchy-aware ontology grounding: many biomedical resources (e.g., HPO, DO, MeSH) organize concepts through deep ``is-a" taxonomies, yet production retrieval stacks overwhelmingly rely on Euclidean embeddings and Euclidean approximate nearest neighbor (ANN) indexes.
Hyperbolic embeddings offer theoretical advantages for hierarchical representation, but face two practical adoption barriers in LLM systems: (i) hyperbolic nearest-neighbor search is not natively supported by most vector databases, and (ii) not every user query depends on hierarchical structure, so methods that exclusively use hyperbolic distance risk underperforming strong Euclidean baselines on entity-centric queries.

We present HyEm, a lightweight retrieval layer that integrates hyperbolic ontology embeddings into existing Euclidean ANN infrastructure without requiring specialized indexing primitives, while maintaining robustness across heterogeneous query types.
HyEm combines three components.
First, we learn hyperbolic entity embeddings under an explicit radius budget, which mitigates numerical instabilities and makes the downstream tangent-space approximation more controllable.
Second, we map Euclidean text embeddings into hyperbolic space via a compact adapter and perform candidate retrieval by storing only origin log-mapped vectors in a standard Euclidean vector database, followed by exact hyperbolic reranking on a small candidate set.
Third, we introduce a query-adaptive gate that outputs a continuous mixing weight, combining Euclidean semantic similarity with hyperbolic hierarchy distance at reranking time to accommodate queries with varying degrees of hierarchy dependence.

Our theoretical analysis builds on a bi-Lipschitz comparison between hyperbolic distance and Euclidean distance in origin normal coordinates under a radius constraint, and translates it into practical guidance for ANN oversampling and embedding dimensionality.
We further analyze how classification errors in query routing affect retrieval quality, motivating the soft mixing design.
Experiments on open biomedical ontology subsets with a stratified query taxonomy demonstrate that HyEm preserves 94--98\% of Euclidean baseline performance on entity-centric queries while substantially improving performance on hierarchy-navigation and mixed-intent queries, with maintained indexability at moderate oversampling factors.
\end{abstract}




\end{frontmatter}


\section{Introduction}
\label{sec:intro}

Large language models (LLMs) have demonstrated capabilities in explaining biomedical concepts, summarizing clinical notes, and answering medical questions, yet their reliability critically depends on grounding mechanisms that connect free-form language to curated knowledge.
Retrieval-augmented generation (RAG) has emerged as a widely adopted engineering pattern: a query is encoded into a vector, a vector database retrieves relevant items, and the LLM generates an answer conditioned on the retrieved evidence.
In biomedicine, a particularly valuable evidence source is a {biomedical ontology}---or ontology-like knowledge graph---in which concepts are organized through deep \texttt{is-a} taxonomies.
This hierarchical organization is consequential in practice: the query ``What are subtypes of {cardiomyopathy}?'' requires fundamentally different retrieval than ``What does {cardiomyopathy} mean?'', despite both mentioning the same surface string.
The former demands navigation of a hierarchy (children/descendants); the latter is closer to entity linking or definition retrieval.

\paragraph{The geometry mismatch has become a deployment issue}
Hyperbolic geometry offers well-studied theoretical advantages for representing trees and taxonomies: hyperbolic volume grows exponentially with radius, mirroring the branching growth of hierarchies.
Foundational theoretical results and empirical evidence suggest that hyperbolic embeddings can represent hierarchies with substantially lower distortion than Euclidean embeddings at comparable dimension \cite{Nickel2017,Sala2018,Ganea2018,Nickel2018,Law2019}.
Despite this promise, hyperbolic retrieval layers remain uncommon in production LLM/RAG stacks, where retrieval is typically implemented as Euclidean/cosine ANN over fixed-dimensional vectors.
The practical question is therefore not only whether hyperbolic embeddings can be effective in controlled settings, but how to integrate them into the dominant Euclidean indexing interface with minimal engineering risk.
In medical software deployments, constraints such as limited backend customization and regulatory requirements can make approaches that rely on specialized hyperbolic indexes or end-to-end retraining harder to adopt.

\paragraph{A second barrier is methodological: hyperbolic is not universally better}
Recent work highlights methodological subtleties in evaluating hyperbolic representations: results can be sensitive to baseline tuning and evaluation protocols, and hyperbolic optimization can suffer from numerical issues and boundary concentration \cite{Li2020,Liu2024,VanNooten2025,Sala2018}.
This critique bears directly on ontology grounding, where many real queries are not hierarchy-navigation queries.
A retrieval system that universally prioritizes hyperbolic distance may regress on ``flat'' semantic similarity queries where cosine retrieval excels.
The question is not whether hyperbolic geometry helps, but when and how to apply it without sacrificing performance on queries that do not depend on hierarchical structure.

\paragraph{Our approach}
This paper presents \hyem{}, a {query-adaptive} hyperbolic retrieval layer engineered to minimize deployment friction.
\hyem{} treats {indexability} as a first-class design constraint: it learns hyperbolic entity embeddings under an explicit radius budget, stores only tangent-space (log-mapped) vectors in a standard Euclidean ANN index, and reranks a small candidate set using exact hyperbolic distance.
To maintain robustness under heterogeneous query intents, \hyem{} incorporates a lightweight gating module that outputs a continuous mixing weight $\alpha(q)$ and interpolates between Euclidean similarity and hyperbolic distance during reranking.
When $\alpha(q) \approx 0$, the system behaves like a Euclidean retriever; when $\alpha(q) \approx 1$, it behaves like a hierarchy-aware hyperbolic retriever.

\paragraph{Contributions}
\hyem{} makes three contributions with an emphasis on deployable design and measurable evaluation.
First, we package a hyperbolic ontology retrieval layer that fits existing Euclidean vector-DB interfaces (tangent-space indexing + exact hyperbolic reranking) and treat \emph{indexability} as an explicit, testable requirement rather than an implicit assumption.
Second, building on standard metric comparisons in normal coordinates, we make the radius--distortion relation explicit and use it to derive practical guidance for ANN oversampling and for the capacity--indexability trade-off induced by ontology depth and branching.
Third, inspired by routing and score-mixing patterns in hybrid retrieval, we introduce query-adaptive geometry (hard routing and soft mixing) and analyze/ablate how it improves hierarchy-navigation queries while preserving entity-centric performance under heterogeneous query streams.

\paragraph{Organization}
Section~\ref{sec:related} reviews background and positions our work relative to hyperbolic embedding, retrieval indexing, and hybrid retrieval.
Section~\ref{sec:setup} formalizes ontology grounding and introduces a query taxonomy that motivates query-adaptive geometry.
Section~\ref{sec:method} presents \hyem{}.
Section~\ref{sec:theory} provides the theoretical analysis.
Section~\ref{sec:experiments} specifies a lightweight experimental protocol and reports results for a reproducibility run.
Appendices provide proofs and concise benchmark/protocol details.

\section{Background and Related Work}
\label{sec:related}

\subsection{Ontology grounding and biomedical knowledge resources}
\label{sec:rw_onto}

Biomedical concept grounding and ontology-based retrieval map free-form mentions (e.g., in questions, clinical notes, or biomedical literature) to canonical identifiers in controlled vocabularies.
Large, expert-curated resources such as the Human Phenotype Ontology (HPO) \cite{Kohler2021HPO}, the Disease Ontology (DO) \cite{Schriml2022DO}, and MeSH \cite{NLMMeSH} provide substantial coverage together with an explicit hierarchical backbone (is-a and related relations), while integrative efforts such as the Monarch Initiative \cite{Mungall2017} and linked-data infrastructures such as Bio2RDF \cite{Callahan2013} connect multiple ontologies and knowledge graphs.
Because ontology concepts are often accessed through names, synonyms, and short textual definitions, modern grounding systems frequently cast candidate generation as a text retrieval problem and rely on contextual encoders derived from Transformer language models \cite{Devlin2019,Lee2019,Gu2021} and sentence-level embedding objectives \cite{Reimers2019}.

In parallel, the rise of large language models (LLMs) has renewed interest in coupling parametric knowledge with explicit symbolic resources.
General-purpose LLM scaling \cite{Brown2020} and reasoning-style prompting \cite{Wei2022} have shown strong few-shot performance, while domain-specialized LLMs target clinical and biomedical tasks \cite{Singhal2022}.
Recent work also explores using LLMs for scientific information extraction and knowledge base construction \cite{Dagdelen2024}.
For ontology engineering specifically, LLMs have been studied as assistants for proposing missing concepts/relations \cite{Zhou2025} and for generating training data via ontology verbalization \cite{Zaitoun2024}; complementary ontology design patterns such as compliance/reward modeling highlight the breadth of structured knowledge needed in real biomedical pipelines \cite{Peleg2024CaRO}.
These trends suggest the potential value of retrieval designs that (i) can exploit hierarchical structure when it is relevant, but (ii) remain robust for ``flat'' similarity queries that do not depend on depth.

\subsection{Hyperbolic representations for hierarchies}
\label{sec:rw_hyp}

Hyperbolic embeddings have become a widely used tool for representing taxonomies following the introduction of Poincar\'{e} embeddings \cite{Nickel2017} and subsequent analyses of representational trade-offs \cite{Sala2018}.
One key motivation is geometric: in negatively curved spaces, volume grows exponentially with radius, which can mirror the node growth in trees and may support lower-distortion embeddings of deep hierarchies \cite{Mettes2024}.
In practice, optimization and numerical stability considerations have driven the use of alternative models of hyperbolic space.
In particular, the Lorentz (hyperboloid) model may improve stability compared to the Poincar\'{e} ball when embeddings approach large radii \cite{Nickel2018,Law2019}, and Riemannian optimization methods can provide principled updates on curved manifolds \cite{Bonnabel2013,Becigneul2019}.

Beyond distance-based objectives, several lines of work encode hierarchy more explicitly.
Partial-order formalisms such as entailment cones aim to represent ancestor/descendant structure via angular constraints \cite{Ganea2018}, and more recent variants further refine these order-theoretic constructions \cite{Yu2024}.
Hyperbolic geometry has also been integrated into neural architectures and representation learning pipelines, including hyperbolic attention mechanisms \cite{Gulcehre2018, Yang2023}, hyperbolic graph neural networks \cite{Zhou2023}, and general tutorials/surveys that systematize gyrovector operations, exp/log maps, and curvature choices \cite{Zhou2023}.
At the lexical level, hyperbolic word embeddings such as Poincar\'{e} GloVe suggest compatibility between distributional semantics and hierarchical organization \cite{Tifrea2019}, while continuous tree and hierarchical clustering objectives aim to connect representation learning with taxonomy induction \cite{Monath2019}.
From the viewpoint of downstream learning, hyperbolic losses and embeddings have been applied to hierarchical prediction problems, including recent evidence that performance gains can depend sensitively on evaluation choices and hierarchy depth \cite{VanNooten2025}.

A closely related thread is knowledge graph (KG) representation learning, where Euclidean embedding families such as TransE/TransR-style translations and bilinear/complex-valued models \cite{Bordes2013, Yang2015, Trouillon2016, Sun2019,Le2023} remain strong baselines and are surveyed in broader relational learning over KGs \cite{Nickel2016}.
Hyperbolic KG models (e.g., MuRP and related approaches) exploit curvature to compress multi-relational graphs with latent hierarchies \cite{Balazevic2019, Chami2019, Chami2020, Zheng2024}, and graph neural networks provide complementary Euclidean relational baselines \cite{Kipf2017, Schlichtkrull2018, Velickovic2018}.
In the biomedical domain, hyperbolic embeddings have been used to encode ontology structure and improve representation quality for clinical codes and gene/phenotype resources \cite{Choi2016,Kim2021}.
More recently, language models have also been studied as hierarchy encoders, suggesting that text encoders can be trained to respect (or reveal) hierarchical organization \cite{He2024}.
\hyem{} builds on these insights by training radius-controlled hyperbolic concept embeddings and exposing both Euclidean and hyperbolic similarity signals for retrieval.

\subsection{When Euclidean baselines match hyperbolic methods}
\label{sec:rw_negative}

The hyperbolic embedding literature reports both strong gains on hierarchy-sensitive tasks \cite{Nickel2017,Sala2018} and negative or mixed evidence in other settings \cite{Dhingra2018}.
A recurring observation is that results can be sensitive to optimization stability, curvature mismatch, and evaluation protocols, and that gains over Euclidean baselines can diminish when Euclidean baselines are strongly tuned \cite{Muscoloni2016, Sala2018,Li2020FL,Liu2024,VanNooten2025}.
This is particularly relevant for ontology grounding because many user queries are not ``about depth'' (e.g., synonymy and sibling disambiguation), and real ontologies often contain cross-links, multiple inheritance, and non-tree relations that may dilute purely hierarchical signals \cite{Nickel2016, Schlichtkrull2018}.
Empirically, even when hierarchies exist, improvements can concentrate in deep, ancestry-sensitive queries, while shallow or noisy regions may not benefit and can even regress if embeddings drift toward the Poincar\'{e} boundary \cite{Nickel2018,Law2019}.

These observations suggest two practical design principles.
First, claims should be appropriately scoped: hyperbolic geometry is not universally superior, but it may be particularly well suited to deep hierarchy navigation and ancestor-aware ranking \cite{Nickel2017,Ganea2018,Yu2024}.
Second, deployment-oriented methods may benefit from mechanisms that avoid harming queries that do not depend on hierarchical structure.
\hyem{} follows both principles by (i) explicitly controlling embedding radius to mitigate boundary issues \cite{Nickel2018,Law2019} and (ii) introducing query-adaptive geometry via soft mixing, which is conceptually aligned with gating and mixture patterns widely used in modern neural systems \cite{Shazeer2017}.

\subsection{Hyperbolic nearest-neighbor search and deployment constraints}
\label{sec:rw_nns}

Efficient nearest-neighbor search (NNS) is important for scalable retrieval, and modern vector databases typically expose Euclidean or cosine similarity over fixed-dimensional vectors, often implemented with graph-based ANN structures and compressed indices \cite{Malkov2020,Johnson2021}.
Hyperbolic NNS has therefore attracted increasing attention as hyperbolic embeddings began to scale beyond toy hierarchies.
Several approaches design bespoke index structures, derive probabilistic bounds, or reduce hyperbolic NNS to repeated Euclidean queries \cite{Prokhorenkova2022}.
Recent work also studies hyperbolic ANN from a more algorithmic perspective \cite{KisfaludiBak2024,Park2025} and proposes compression/indexing schemes tailored to curved geometry, such as hyperbolic product quantization \cite{Qiu2024}.

While valuable, existing hyperbolic NNS methods often assume either specialized index structures or the ability to deploy non-Euclidean distance functions in the retrieval backend.
In many LLM/RAG deployments, retrieval is exposed as a relatively fixed Euclidean/cosine vector-DB API: one can insert vectors, query by vectors, and optionally rerank a modest candidate set.
A common compromise in the hyperbolic literature is therefore to leverage tangent-space (log-map) representations, which make local neighborhoods approximately Euclidean \cite{Nickel2018,Gulcehre2018,Yang2023,Zhou2023}.
\hyem{} builds on this compromise: it issues Euclidean ANN search on log-mapped vectors and then performs exact hyperbolic reranking on a small candidate set.
Compared with generic tangent-space pipelines, \hyem{} makes two aspects explicit for ontology grounding: (i) a radius budget that keeps tangent-space distortion controlled, and (ii) an empirical indexability stress test that quantifies how much ANN oversampling is needed before reranking recovers the true hyperbolic top-$k$ neighbors.

\subsection{Hybrid retrieval and query routing in RAG systems}
\label{sec:rw_router}

Retrieval-augmented generation (RAG) pipelines increasingly combine heterogeneous retrievers and scoring functions, mixing dense neural retrieval with lexical retrieval and reranking. Classic sparse baselines such as BM25 remain competitive and widely used in practice \cite{Robertson2009}, while dense dual-encoder retrievers such as DPR popularized end-to-end learned retrieval for open-domain QA \cite{Karpukhin2020}.
RAG models then couple such retrieval with generation \cite{Lewis2020RAG}, and downstream systems often strengthen evidence aggregation with multi-passage generative readers \cite{Izacard2021} or retrieval-augmented LMs that are trained with retrieval in the loop \cite{Izacard2022,Gao2022HyDE}.
Architectures that improve retrieval efficiency via lightweight reranking or late interaction (e.g., token-level matching) further blur the boundary between retrieval and scoring \cite{Khattab2020}.

A common pattern in production RAG is to use query classification or learned routing to decide which retriever (or which mixture of scores) should dominate for a given query.
Recent ``selective RAG'' and routing work explicitly trains models to decide when to retrieve, or which retrieval-augmented model to use, based on the query and intermediate signals \cite{Asai2024}.
At the same time, hyperbolic geometry has begun to appear inside graph-based RAG variants as a way to represent abstraction depth and hierarchical containment \cite{He2025}.
More broadly, hybrid retrieval often combines multiple similarity signals, but those signals are usually defined within a single geometry (most commonly Euclidean/cosine).
\hyem{} adapts routing and score mixing to a \emph{geometry} choice: it treats Euclidean similarity and hyperbolic distance as complementary signals over the {same} ontology store, learns a query-dependent mixing weight $\alpha(q)$, and constrains the hyperbolic component via radius-controlled training so that it remains compatible with Euclidean ANN indexing.

\section{Problem Setup}
\label{sec:setup}

\subsection{Ontology grounding as retrieval}
\label{sec:setup_task}

We consider an ontology or ontology-like biomedical knowledge graph $G=(V,E)$.
Each node $v\in V$ denotes a standardized biomedical entity (phenotype, disease, procedure).
We assume the graph is primarily structured by an \texttt{is-a} relation that defines a directed acyclic graph (DAG) $E_{\prec}\subseteq E$, and that each node has a short text field $\tau(v)$ consisting of a preferred label plus optional synonyms/definition.
A retrieval module receives a natural-language query $q$ and returns a ranked list of entities $\pi(q)=(v_1,\ldots,v_k)$.
In an LLM system, the retrieved entities can be used as tool outputs (entity normalization), as context expansion (retrieving parents/children/definitions), or as structured constraints.
This paper focuses on evaluating the retrieval module itself, independent of downstream generation.

\subsection{A query taxonomy for ontology grounding}
\label{sec:taxonomy}

Ontology grounding queries exhibit considerable heterogeneity.
We identify three recurring query families, which we later use to construct lightweight benchmark sets.

\paragraph{Entity-centric queries}
These queries primarily ask ``what is $X$?'' or provide a synonym/description that should map to a single entity.
They can often be effectively addressed using Euclidean semantic similarity (cosine) between text embeddings.

\paragraph{Taxonomy-navigation queries}
These queries explicitly request hierarchical navigation, for example ``subtypes of $X$'' (children/descendants), ``broader category of $X$'' (parents/ancestors), or ``classification path from $A$ to $B$''.
Here, correctness is primarily determined by the \texttt{is-a} structure rather than by flat semantic similarity.

\paragraph{Mixed-intent queries}
These queries involve both semantic similarity and hierarchical constraints.
An example is ``diseases related to $X$ at the same specificity level'', which combines similarity with a depth constraint.
Mixed-intent queries suggest the potential value of soft mixing: rather than relying on a binary choice, the method can interpolate between Euclidean and hyperbolic signals.

\subsection{Notation}
\label{sec:notation}

Let $f_{\text{text}}(\cdot)$ be a sentence encoder that maps text to Euclidean vectors $e\in\R^{d_e}$.
We denote entity text embeddings by $e_v=f_{\text{text}}(\tau(v))$ and query embeddings by $e_q=f_{\text{text}}(q)$.
\hyem{} additionally learns hyperbolic entity embeddings $x_v\in\Hspace^{d}$ and a lightweight adapter $g$ that maps $e_q$ into $x_q=g(e_q)$.

\section{\hyem{}: Query-Adaptive, Indexable Hyperbolic Retrieval}
\label{sec:method}

\subsection{Design goals and overview}
\label{sec:design_goals}

\hyem{} is designed around three goals.
First, the method should be compatible with existing Euclidean vector databases and ANN engines, so that deployment does not require new indexing primitives.
Second, the method aims to preserve the potential hierarchy-aware benefits of hyperbolic geometry on taxonomy-navigation queries.
Third, because real query streams are heterogeneous, the method aims to avoid negatively affecting entity-centric queries where Euclidean retrieval is already strong.

Figure~\ref{fig:pipeline} illustrates \hyem{}.
Entity embeddings are trained in hyperbolic space under a radius constraint.
At indexing time, we store only their origin log-mapped vectors in a Euclidean ANN index.
At query time, we compute both a Euclidean text embedding (for semantic similarity) and a hyperbolic query embedding (for hierarchy-aware distance).
A lightweight gate outputs $\alpha(q)\in[0,1]$ and controls a mixed reranking score.

\begin{figure}[t]
  \centering
  \includegraphics[width=\linewidth]{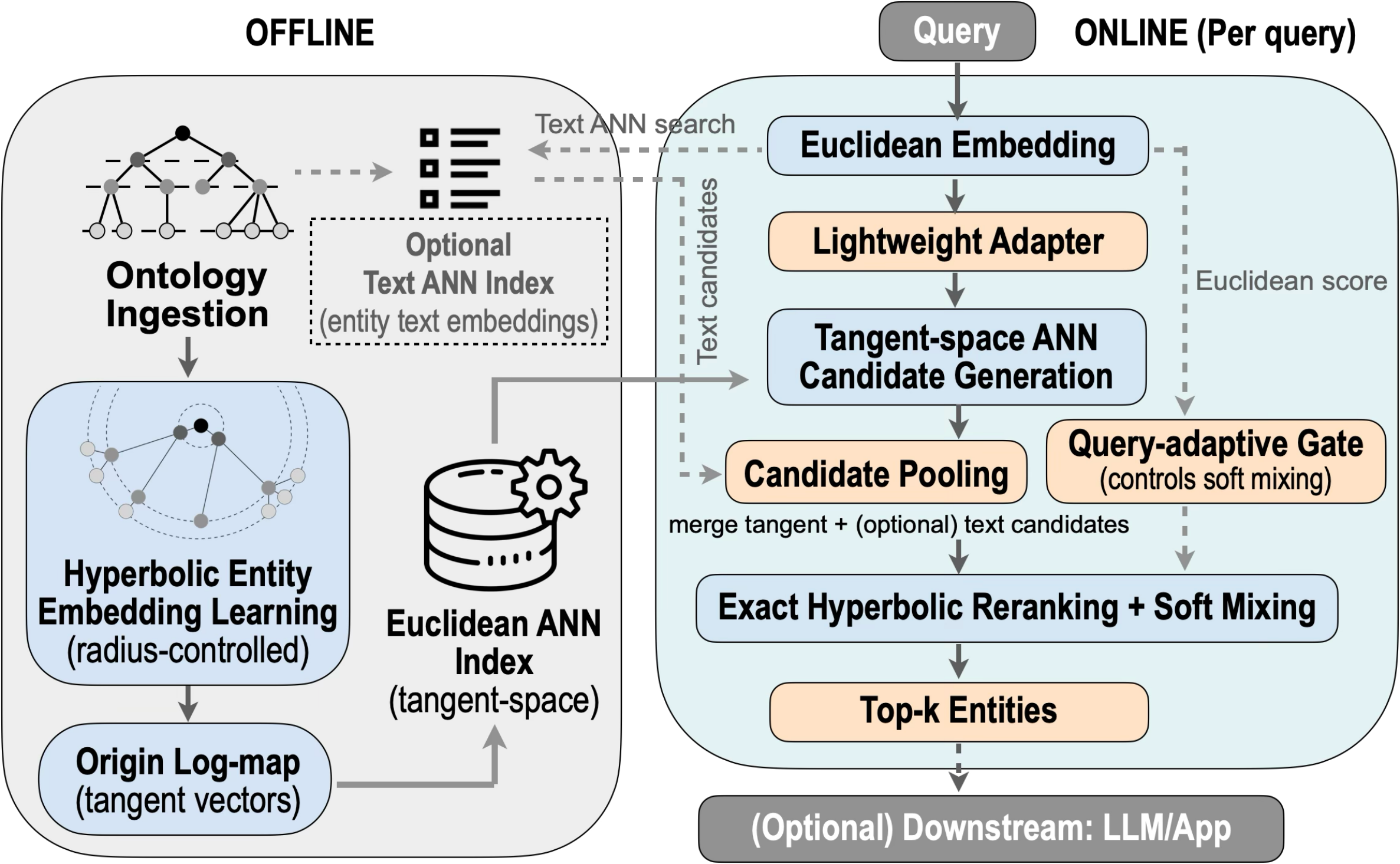}
\caption{Overall HyEm pipeline. Offline, HyEm ingests an ontology, learns radius-controlled hyperbolic entity embeddings, applies an origin log-map to obtain tangent vectors, and builds a standard Euclidean ANN index in the tangent space. Optionally, HyEm also builds a text ANN index over entity text embeddings. Online (per query), the query is encoded into a Euclidean embedding and mapped by a lightweight adapter for tangent-space ANN candidate generation. Candidate pooling merges tangent candidates with optional text candidates, followed by exact hyperbolic reranking with a query-adaptive gate that controls soft mixing between hyperbolic and Euclidean scores. The output is the top-k retrieved entities, optionally consumed by a downstream LLM/application.}
  \label{fig:pipeline}
\end{figure}

\paragraph{Offline vs. online paths}
Figure~\ref{fig:pipeline} separates two phases.
\textbf{Offline}: we learn hyperbolic entity representations and store their tangent vectors $u_v=\log_0(x_v)$ in a standard Euclidean ANN index;
this is where we interact with the vector database.
\textbf{Online}: a text query is embedded in Euclidean space, mapped by a lightweight adapter into the same tangent space, and then used for ANN search.
Importantly, the logarithmic map is not applied to queries at runtime---the adapter outputs tangent coordinates directly---so the ``hyperbolic'' component
appears only as (i) a different training geometry and (ii) a reranking distance computed on a small candidate set.

\paragraph{End-to-end retrieval pipeline}
Algorithm~\ref{alg:hyem_index} summarizes the indexing and query-time steps, including optional candidate pooling and soft mixing reranking.

\begin{algorithm}[t]
\small
\caption{\hyem{}: tangent-space indexing, optional candidate pooling, and soft mixing reranking}
\label{alg:hyem_index}
\begin{algorithmic}[1]
\Require Ontology nodes $V$ with text $\tau(v)$ and hierarchy edges $E_{\prec}$; 
Text encoder $f_{\text{text}}$ (frozen or lightly tuned);
Trainable adapter $g(e)=\exp_0(We+b)$;
Optional gate $\alpha(q)=\sigma(w_g^\top e_q+b_g)$.
\State \textbf{Train entity embeddings:} optimize $\{x_v\}$ and adapter parameters by minimizing $\mathcal{L}_{\text{hier}}+\lambda_{\text{text}}\mathcal{L}_{\text{text}}+\lambda_R\mathcal{L}_R$.
\State \textbf{Build indexes:}
\State \quad compute $u_v\leftarrow \log_0(x_v)$ for all $v$ and build Euclidean ANN index $\mathcal{I}_H$ on $\{u_v\}$.
\State \quad compute text vectors $e_v\leftarrow f_{\text{text}}(\tau(v))$ and build Euclidean ANN index $\mathcal{I}_E$ on $\{e_v\}$ (cosine or inner product).
\State \textbf{Query($q,k,L_H,L_E$):}
\State \quad $e_q\leftarrow f_{\text{text}}(q)$
\State \quad $x_q\leftarrow g(e_q)$ and $u_q\leftarrow \log_0(x_q)$
\State \quad $\alpha\leftarrow \alpha(q)$ \Comment{If no gate, set $\alpha\leftarrow 1$ or tune globally.}
\State \quad $C_H\leftarrow \mathrm{ANN}(\mathcal{I}_H,u_q,L_H)$
\State \quad $C_E\leftarrow \mathrm{ANN}(\mathcal{I}_E,e_q,L_E)$ \Comment{Optional; can be skipped for latency.}
\State \quad $C\leftarrow C_H\cup C_E$
\State \quad for $v\in C$: compute $s_H\leftarrow -d_{\Hspace}(x_q,x_v)$ and $s_E\leftarrow \cos(e_q,e_v)$
\State \quad score$(v)\leftarrow \alpha\,s_H+(1-\alpha)\,s_E$ \Comment{Optionally use temperature scaling.}
\State \quad \textbf{return} top-$k$ entities by score
\end{algorithmic}
\end{algorithm}

\subsection{Geometry choice: Lorentz implementation with model-agnostic analysis}
\label{sec:geom}

The theoretical analysis in Section~\ref{sec:theory} is model-agnostic: it uses intrinsic hyperbolic geometry and normal coordinates.
For implementation, however, numerical stability matters.
We therefore emphasize the Lorentz (hyperboloid) model during optimization, following prior observations that it reduces instabilities near large radii \cite{Nickel2018,Law2019}.
Concretely, we represent each point as $y\in\R^{d+1}$ on the hyperboloid $\langle y,y\rangle_L=-1$ with $y_0>0$ under the Lorentzian inner product.
Distances are computed as $d_{\Hspace}(y_1,y_2)=\operatorname{arcosh}(-\langle y_1,y_2\rangle_L)$.
We use standard Lorentz log/exp maps for optimization (\ref{app:math}). 
When convenient for exposition, we use the Poincar\'e ball notation $x\in\Ball^d$; the two parameterizations are isometric.

\subsection{Radius-constrained hyperbolic entity embeddings}
\label{sec:entity_embed}

Let $x_v\in\Hspace^d$ denote the hyperbolic embedding of entity $v$.
Unconstrained hyperbolic training can push deep nodes toward the boundary (in the Poincar\'e picture) or to very large norms (in Lorentz coordinates), which may amplify metric distortion when one later approximates hyperbolic neighborhoods in tangent space.
Because \hyem{} relies on tangent-space ANN indexing, we treat the maximum radius as a controllable resource.

\paragraph{Radius budget}
We introduce a radius budget $R>0$ defined in origin normal coordinates.
Writing $u_v=\log_0(x_v)$ for the origin log map, we enforce $\lVert u_v\rVert\le R$ for all entities.
This budget serves two purposes.
It aims to prevent numerical instabilities associated with extreme radii, and it bounds the multiplicative distortion between hyperbolic distance and Euclidean distance in tangent space (Theorem~\ref{thm:distortion}).

\paragraph{Training objective}
We learn $\{x_v\}$ using a lightweight combination of hierarchy supervision and text alignment.
The hierarchy term is designed to encourage parent--child proximity and monotonic increase of radius with depth.
The text term aims to align the hyperbolic embedding with a frozen (or lightly tuned) text encoder, with the goal that entity retrieval can remain feasible from natural language.
The overall loss is
\begin{equation}
\mathcal{L}=\mathcal{L}_{\text{hier}} + \lambda_{\text{text}}\mathcal{L}_{\text{text}} + \lambda_R\mathcal{L}_R.
\end{equation}
We deliberately keep these losses simple to facilitate reproducibility.
\ref{app:exp} specifies a concrete recipe that can be implemented with standard Python toolkits. 

\paragraph{Hierarchy term}
For each \texttt{is-a} edge $(p\prec c)\in E_{\prec}$, we penalize large hyperbolic distance and enforce the expected radial ordering by a margin.
We additionally use negative sampling to separate children from non-descendants.
This provides a ranking-style loss that is computationally efficient.
While partial-order formalisms (e.g., cones) are compatible with \hyem{}, we treat them as optional extensions rather than core dependencies.

\paragraph{Radius penalty}
We implement the radius budget as a soft penalty $\mathcal{L}_R=\sum_v \max(0,\lVert u_v\rVert-R)^2$ and, optionally, as explicit clipping in tangent space after each update.
The soft penalty is differentiable and works in both Lorentz and Poincar\'e parameterizations.

\subsection{A lightweight Euclidean-to-hyperbolic adapter}
\label{sec:adapter}

Given a query text $q$, a sentence encoder produces a Euclidean embedding $e_q\in\R^{d_e}$.
\hyem{} maps it into hyperbolic space via a small adapter $g$ defined in origin tangent space:
\begin{equation}
  x_q=g(e_q)=\exp_0\big(We_q+b\big),
\end{equation}
where $W\in\R^{d\times d_e}$ and $b\in\R^{d}$ are trainable parameters and $\exp_0$ is the hyperbolic exponential map.
This adapter is deliberately kept lightweight: it can be trained on automatically generated query--entity pairs derived from ontology synonyms and hierarchy templates, without requiring any LLM finetuning.
In ablations we also consider a two-layer MLP in tangent space, but our experiments suggest that a linear adapter is often sufficient.

\subsection{Indexable retrieval: tangent-space candidates and hyperbolic reranking}
\label{sec:retrieval}

\hyem{} aims to make hyperbolic retrieval compatible with Euclidean ANN by indexing only tangent vectors.
For each entity, we store $u_v=\log_0(x_v)\in\R^d$ in a standard Euclidean vector database.
At query time, we compute $x_q=g(e_q)$ and $u_q=\log_0(x_q)$.
Candidate retrieval is a single Euclidean ANN query that returns the top-$L$ nearest tangent vectors to $u_q$.
We then rerank those candidates by exact hyperbolic distance $d_{\Hspace}(x_q,x_v)$.
The radius budget is designed to control how well this candidate set approximates the true hyperbolic neighborhood.

This design offers a practical advantage: only the reranker needs hyperbolic distance computations, and it runs on at most a few hundred candidates.
The vector database itself remains entirely Euclidean.

\subsection{Query-adaptive geometry via soft mixing}
\label{sec:gating}

A core motivation for \hyem{} is that real query streams are heterogeneous.
For entity-centric queries, Euclidean cosine similarity can be a strong signal.
For taxonomy-navigation queries, hyperbolic distance may provide an inductive bias for depth separation and subtree structure.
Rather than relying on a binary choice, we introduce a lightweight gate that outputs a continuous mixing weight $\alpha(q)\in[0,1]$.

\paragraph{Gate architecture and training}
We use a minimal logistic gate on the Euclidean query embedding:
\begin{equation}
\alpha(q)=\sigma(w^\top e_q + b),
\end{equation}
where $\sigma$ is the sigmoid.
{Training is a separate lightweight step}: after training the entity embeddings and adapter, we freeze them and fit $(w,b)$ as a binary classifier on automatically labeled Q-H vs.~Q-E queries using standard cross-entropy.
In our reproducible benchmark, labels come directly from the query templates (Section~\ref{sec:exp_taxonomy}), so no manual annotation is required.
Because the gate only controls score mixing at reranking time (and does not affect candidate generation), this decoupled training approach has proven sufficient in our experiments; end-to-end joint training is an optional extension.

\paragraph{Mixed reranking score}
For each candidate entity $v$ we compute two scores.
The Euclidean score is semantic similarity $s_E(v|q)=\cos(e_q,e_v)$.
The hyperbolic score is the negative geodesic distance $s_H(v|q)=-d_{\Hspace}(x_q,x_v)$.
We then combine them as
\begin{equation}
\text{score}(v|q)=\alpha(q)\,s_H(v|q) + (1-\alpha(q))\,s_E(v|q).
\label{eq:mix_score}
\end{equation}
When $\alpha(q)$ is near one, reranking behaves like hyperbolic retrieval; when it is near zero, reranking behaves like Euclidean retrieval.
This ``soft'' design aims to avoid hard failures caused by misclassification and to address concerns that hyperbolic methods can be brittle under strong Euclidean baselines.

\paragraph{Candidate pooling}
Mixing scores is only useful if the correct entity appears in the candidate pool.
In experiments, we therefore consider two candidate sources: $C_H$ from tangent-space ANN on $u_v$ and $C_E$ from Euclidean ANN on $e_v$.
We take the union $C=C_H\cup C_E$ and rerank using Eq.~\eqref{eq:mix_score}.
This keeps engineering simple (two ANN queries) and is designed to make the method robust even when one candidate generator is imperfect.
In latency-sensitive settings, one may skip $C_E$ or make $L_E$ query-dependent; we treat this as an implementation choice rather than a modeling novelty.

\section{Theory: Indexability, Capacity, and Routing Robustness}
\label{sec:theory}

This section provides explicit bounds that connect geometry to retrieval engineering choices.
Proofs and additional mathematical details are in \ref{app:math}. 

\subsection{A bi-Lipschitz bound for tangent-space indexing}
\label{sec:bilyps}

\hyem{} indexes $u_v=\log_0(x_v)$ in a Euclidean ANN engine.
A key question is when Euclidean distances between tangent vectors preserve the ranking induced by true hyperbolic distances.
The answer depends on how far from the origin embeddings are allowed to move.

\begin{theorem}[Tangent-space distortion under radius $R$]
\label{thm:distortion}
Let $x=\exp_0(u)$ and $y=\exp_0(v)$ with $\lVert u\rVert,\lVert v\rVert\le R$.
Then the hyperbolic distance satisfies
\begin{equation}
  \lVert u-v\rVert \;\le\; d_{\Hspace}(x,y) \;\le\; \kappa(R)\,\lVert u-v\rVert,
  \qquad \text{where } \kappa(R)=\frac{\sinh(R)}{R}.
\end{equation}
\end{theorem}

Theorem~\ref{thm:distortion} is a standard comparison between the hyperbolic metric and the Euclidean metric in normal coordinates; we restate it to make the radius--indexability connection explicit for retrieval engineering.
When $R$ is small, $\kappa(R)$ is close to one and tangent-space ANN provides an accurate approximation to hyperbolic neighborhoods.
As $R$ grows, distortion increases approximately as $e^{R}/(2R)$.
This relationship motivates the design choice to constrain radii during training.

\subsection{From distortion to a candidate oversampling rule}
\label{sec:oversample}

Candidate retrieval in \hyem{} retrieves the top-$L$ Euclidean neighbors in tangent space and reranks them by exact hyperbolic distance.
Because tangent-space distances can distort the true hyperbolic ordering, we typically set $L\gg k$.
Theorem~\ref{thm:distortion} implies that tangent balls map to hyperbolic balls up to a multiplicative expansion by $\kappa(R)$.
This yields a practical sufficient condition for rank preservation.

\begin{proposition}[Rank stability under a multiplicative gap]
\label{prop:rank}
Fix a query point $x_q$ and assume $\lVert\log_0(x_v)\rVert\le R$ for all entities.
Let $v_{(k)}$ be the $k$-th nearest neighbor of $x_q$ under hyperbolic distance and let $d_k=d_{\Hspace}(x_q,x_{v_{(k)}})$.
If every non-top-$k$ entity $v$ satisfies $d_{\Hspace}(x_q,x_v) > \kappa(R)\, d_k$, then the hyperbolic top-$k$ set equals the Euclidean top-$k$ set in tangent space.
\end{proposition}

The strict gap condition is rarely met uniformly, but Proposition~\ref{prop:rank} suggests an engineering heuristic: keep $R$ moderate and use an oversampling factor $L/k$ that increases with $\kappa(R)$.
Section~\ref{sec:exp_indexability} proposes a stress test that empirically measures recall@$k$ as a function of $L$ and $R$.

\subsection{Capacity--indexability trade-off}
\label{sec:capacity}

Constraining radii improves indexability, but it may reduce how well a hierarchy can be represented.
A standard way to quantify this trade-off is via hyperbolic volume growth: hyperbolic space supports exponentially many separated points within radius $R$, roughly $\exp((d-1)R)$.
We use this fact to obtain a simple lower bound connecting ontology size/depth to the necessary radius.

\begin{proposition}[Radius needed for a $b$-ary hierarchy]
\label{prop:capacity}
Consider a $b$-ary tree of depth $D$ (about $b^D$ leaves).
Any embedding that assigns distinct leaf nodes to points with minimum pairwise hyperbolic separation $\ge \varepsilon$ must have maximum radius
\begin{equation}
  R \;\gtrsim\; \frac{D\log b}{d-1} - O\!\left(\frac{\log(1/\varepsilon)}{d-1}\right).
\end{equation}
\end{proposition}

\noindent\textbf{Intuition and practical use.}
Proposition~\ref{prop:capacity} is essentially a capacity statement: a depth-$D$ $b$-ary hierarchy has about $b^D$ leaves,
so any embedding that keeps distinct leaves separated must be able to ``fit'' exponentially many points.
Hyperbolic space can accommodate this with {radius} that grows only linearly in $D$ (and only logarithmically in the number of leaves),
but the required radius still increases as the hierarchy becomes deeper or bushier.
In practice we interpret the bound as a {starting rule of thumb}: given an observed depth $D$ and an empirical branching factor $b$,
we choose the smallest radius budget $R$ that (i) is large enough for the ontology to ``fit'' (capacity) while (ii) keeping the tangent-space
distortion factor $\kappa(R)$ (Theorem~\ref{thm:distortion}) modest so that ANN over $\log_0(x_v)$ remains accurate.
This theory-guided tension explains why \hyem{} treats $R$ as a first-class hyperparameter rather than a hidden training artifact.

Proposition~\ref{prop:capacity} provides a principled starting point for choosing $(d,R)$ given an ontology's observed depth $D$ and branching factor $b$.
Together with Theorem~\ref{thm:distortion}, it makes explicit the design trade-off: a larger $R$ increases representational capacity but increases tangent-space distortion.

\paragraph{A scale sanity check (depth $D\approx 30$)}
A natural question is whether for large biomedical ontologies (e.g., SNOMED-style hierarchies with depth $D>30$) the required radius $R$ might become so large that $\kappa(R)$ grows impractically.
Instantiating Proposition~\ref{prop:capacity} with $d=32$ and branching factors $b\in[2,10]$ gives $R\approx 0.7$--$2.2$ at $D=30$, which corresponds to $\kappa(R)\approx 1.08$--$2.06$.
This suggests that the ``radius explosion'' regime (e.g., $R\approx 10$) is not inevitable at depth 30 when $d$ is moderate; it becomes relevant for extremely deep/bushy hierarchies or for very low-dimensional embeddings.
Section~\ref{sec:discussion} discusses practical considerations and scale-up experiments.

\subsection{Routing robustness: why soft mixing is safer than hard decisions}
\label{sec:routing_theory}

The query-adaptive gate in \hyem{} can be used in two modes.
A hard router uses a threshold on $\alpha(q)$ to choose either Euclidean or hyperbolic retrieval.
A soft router uses $\alpha(q)$ only in the mixed score (Eq.~\eqref{eq:mix_score}), avoiding a discrete binary switch.

We formalize the properties of soft mixing with a simple risk decomposition.
Let $\ell(q,\pi)$ be any retrieval loss (e.g., 1$-$Hits@$k$) for query $q$ under ranking $\pi$.
Let $\pi_E$ and $\pi_H$ denote rankings produced by pure Euclidean and pure hyperbolic scoring, respectively.
Define the per-query regret of choosing $E$ when $H$ is better as $\Delta_H(q)=\ell(q,\pi_E)-\ell(q,\pi_H)$, and similarly $\Delta_E(q)=\ell(q,\pi_H)-\ell(q,\pi_E)$ when $E$ is better.

\begin{proposition}[Expected loss under a hard router]
\label{prop:router}
Assume queries are drawn from a distribution that mixes two latent intents: hierarchy-navigation ($Z=H$) and entity-centric ($Z=E$).
Let a hard router choose hyperbolic when it predicts $\hat Z=H$.
Let $\varepsilon_H=\Pr[\hat Z=E\mid Z=H]$ and $\varepsilon_E=\Pr[\hat Z=H\mid Z=E]$ be the false-negative and false-positive rates.
Then the expected loss of the routed system satisfies
\begin{equation}
\begin{aligned}
\mathbb{E}[\ell(q,\pi_{\text{route}})]
&= \mathbb{E}[\min\{\ell(q,\pi_E),\ell(q,\pi_H)\}] \\
&\quad + \mathbb{E}[\Delta_H(q)\,\mathbf{1}\{Z=H,\hat Z=E\}] \\
&\quad + \mathbb{E}[\Delta_E(q)\,\mathbf{1}\{Z=E,\hat Z=H\}].
\end{aligned}
\end{equation}
\end{proposition}

Proposition~\ref{prop:router} shows that the performance of routing depends on its classification errors: the gap to an oracle (which always picks the better geometry) is determined by misrouting probabilities weighted by the cost of misrouting.
Soft mixing is designed to reduce this cost by replacing the discrete choice with interpolation, so that even a miscalibrated gate is less likely to behave like the fully wrong retriever.
In experiments we quantify this effect by comparing hard routing against soft mixing with the same gate.

\section{Experiments}
\label{sec:experiments}

Our goal is to design experiments that are both publishable and lightweight, allowing full reproduction in a small Python repository.
We therefore focus on open biomedical ontologies and on retrieval tasks that do not require patient data.
We report a CPU-only reproducibility run on two 5k subsets (HPO-5k and DO-5k, seed=0). 
To address scaling concerns, we additionally report a scale-up run on deterministic 20k-node subsets (HPO-20k and DO-20k, seed=0) under the same hyperparameters (Table~\ref{tab:scale_20k}).
Implementation details are in \ref{app:exp}. 

\subsection{Datasets: open biomedical ontologies}
\label{sec:datasets}

We evaluate on two widely used, openly accessible biomedical ontology resources.
The Human Phenotype Ontology (HPO) provides a deep phenotype taxonomy with rich synonymy \cite{Kohler2021HPO}.
The Disease Ontology (DO) provides disease concepts and a curated \texttt{is-a} backbone \cite{Schriml2022DO}.
MeSH is supported by the pipeline as an optional extension \cite{NLMMeSH}, but is omitted from the main results here to keep the reproducibility run small.

To keep experiments scalable and fully reproducible, we report results on size-controlled 5k subsets constructed by sampling subtrees while preserving depth statistics (seed=0).
Table~\ref{tab:data_stats} summarizes node counts, edge counts, and depth statistics for the subsets used in this version.

\begin{table}[t]
\centering
\footnotesize
\caption{Dataset statistics for the 5k subsets used in our reproducibility run (seed=0).}
\label{tab:data_stats}
\begin{tabular}{lcccc}
\toprule
Dataset & \#Nodes & \#\texttt{is-a} edges & Max depth & Avg. branching \\
\midrule
HPO-5k & 5000 & 5523 & 6 & 1.10 \\
DO-5k & 5000 & 7298 & 5 & 1.46 \\
\bottomrule
\end{tabular}
\end{table}

\subsection{Benchmark construction and query taxonomy}
\label{sec:exp_taxonomy}

A key observation from recent critiques is that aggregate metrics can obscure whether a method benefits the intended query types.
We therefore construct a query set that explicitly instantiates the taxonomy of Section~\ref{sec:taxonomy}.
All queries are generated automatically from ontology structure and text fields, enabling full reproducibility.

\paragraph{Entity-centric queries (Q-E)}
We generate queries by sampling synonyms and short definitions for each entity.
Each query has a single ground-truth target entity.
These queries are designed to emphasize semantic matching rather than hierarchy navigation.

\paragraph{Taxonomy-navigation queries (Q-H)}
We generate queries that ask for parents/ancestors, children/descendants, or classification paths.
Ground truth is derived directly from the \texttt{is-a} graph.
These queries are designed to evaluate whether a method preserves hierarchical structure.

\paragraph{Mixed-intent queries (Q-M)}
We generate ``same-level'' or sibling-style queries that ask for concepts similar to $v$ but constrained to a comparable specificity level.
Operationally, we use the node's depth bucket as a proxy for specificity and treat sibling sets as ground truth.
These queries provide the primary motivation for soft mixing.

\paragraph{Gate training labels}
We train the query gate using automatically assigned labels.
Q-H queries are labeled as hierarchy-navigation; Q-E queries are labeled as entity-centric.
Q-M queries are used either as a third category (for analysis) or as unlabeled test queries, depending on the ablation.
This approach keeps the gate lightweight and avoids subjective manual labeling.
We report gate accuracy and AUC alongside end-to-end retrieval metrics.

\subsection{Compared methods and baselines}
\label{sec:baselines}

We compare \hyem{} against Euclidean and hyperbolic baselines that are feasible in a lightweight codebase.
The primary Euclidean baseline is standard text embedding retrieval using cosine similarity over $e_v$.
We also include a simple Euclidean graph-embedding baseline trained on \texttt{is-a} edges (e.g., TransE/DistMult style) to test whether graph structure alone in Euclidean space is sufficient.
On the hyperbolic side, we include a version without radius control to quantify the effect of the indexability constraint.
Finally, we evaluate \hyem{} with no gating, with hard routing, and with soft mixing.

\subsection{Implementation details}
\label{sec:impl}

We implement all methods in a single, lightweight Python codebase and release end-to-end scripts that download data, build subsets, train models, build indexes, and reproduce every number/plot in this paper.
For text representations we use a SentenceTransformer encoder \cite{Reimers2019}, and for ANN we use HNSW-style vector search \cite{Malkov2020} (the same ANN family as widely deployed libraries such as FAISS \cite{Johnson2021}).
Unless otherwise stated, we report results for a single deterministic subset seed (seed=0) on HPO-5k and DO-5k, so that the full pipeline can be run on commodity hardware.
The codebase supports multiple seeds and larger subsets; multi-seed aggregation and significance testing are left as optional extensions.
The exact protocol (including seeds, subset construction, and metric implementations) is documented in \ref{app:exp} and in the released repository. 

\paragraph{Text encoder and synonym indexing}
We use a sentence encoder that is easy to run and reproduce for the core structural experiments (Q-H/Q-M, indexability, and scaling).
However, a minimal Q-E evaluation would be uninformative for assessing whether \hyem{} maintains performance on entity-centric retrieval.
We therefore evaluate Q-E in a separate {realistic} entity-normalization setting (Table~\ref{tab:qe}) by (i) indexing at least one synonym per entity (in addition to the preferred label and one definition sentence) and (ii) using a biomedical sentence encoder (S-BioBERT) to obtain a competitive Euclidean baseline.

\paragraph{Hyperbolic training}
We train entity embeddings in the Lorentz model for numerical stability and convert to/from tangent coordinates as needed.
We fix curvature to a constant (e.g., $-1$) for the main experiments and treat learned curvature as an optional ablation.
We tune the radius budget $R$ and embedding dimension $d$ using a small validation set guided by Proposition~\ref{prop:capacity}.

\paragraph{Indexes and oversampling}
We build Euclidean ANN indexes using a standard library (FAISS or HNSW).
We measure how oversampling $L$ affects recall of the true hyperbolic top-$k$ neighbors, and we relate the observed curves to Theorem~\ref{thm:distortion}.

\subsection{Main results}
\label{sec:main_results}

We report results separately for Q-E, Q-H, and Q-M queries.
This stratification is important: a method can appear strong overall while underperforming on hierarchy navigation, or vice versa.

\begin{table}[t]
\centering
\footnotesize
\caption{Entity-centric retrieval (Q-E) in a realistic entity-normalization setting (5k subsets, seed=0): biomedical sentence encoder (S-BioBERT) and indexing one synonym per entity. The ``Retention'' column shows the MRR ratio relative to the Euclidean text baseline, quantifying the safety-valve effect.}
\label{tab:qe}
\begin{tabular}{llcccc}
\toprule
Dataset & Method & Hits@1 & Hits@10 & MRR & Retention \\
\midrule
\multirow{5}{*}{HPO-5k}
 & Euclidean text retrieval & 0.816 & 0.969 & 0.869 & 1.000 \\
 & Euclidean KG embedding & 0.009 & 0.048 & 0.021 & 0.024 \\
 & Hyperbolic (no radius ctrl) & 0.009 & 0.048 & 0.020 & 0.023 \\
 & \hyem{} (no gate) & 0.009 & 0.045 & 0.016 & 0.019 \\
 & \hyem{} (soft mix) & 0.736 & 0.952 & 0.815 & \textbf{0.939} \\
\midrule
\multirow{5}{*}{DO-5k}
 & Euclidean text retrieval & 0.618 & 0.834 & 0.688 & 1.000 \\
 & Euclidean KG embedding & 0.006 & 0.054 & 0.016 & 0.023 \\
 & Hyperbolic (no radius ctrl) & 0.008 & 0.044 & 0.017 & 0.025 \\
 & \hyem{} (no gate) & 0.012 & 0.045 & 0.019 & 0.028 \\
 & \hyem{} (soft mix) & 0.598 & 0.821 & 0.673 & \textbf{0.979} \\
\bottomrule
\end{tabular}
\end{table}

\paragraph{Q-E (entity-centric): strong-baseline check}
Table~\ref{tab:qe} is designed to evaluate whether \hyem{} maintains performance on entity-centric retrieval when the Euclidean baseline is strong.
We follow a standard entity-normalization practice by indexing at least one synonym per entity and using a biomedical sentence encoder.
The ``Retention'' column quantifies the safety-valve effect: \hyem{} (soft mix) preserves 93.9\% (HPO-5k) and 97.9\% (DO-5k) of the Euclidean baseline MRR, while simultaneously showing substantial improvements on Q-H queries (see Table~\ref{tab:qh}).
This indicates that soft mixing can provide a safety valve: when the gate assigns low hierarchy weight (small $\alpha(q)$), reranking interpolates toward Euclidean similarity while retaining the ability to exploit hyperbolic structure for hierarchy-sensitive queries.

\begin{table}[t]
\centering
\footnotesize
\caption{Taxonomy-navigation retrieval (Q-H) on 5k ontology subsets (seed=0).}
\label{tab:qh}
\begin{tabular}{llcccc}
\toprule
Dataset & Method & \shortstack{Parent\\Hits@5} & \shortstack{Parent\\Hits@10} & \shortstack{Ancestor F1\\(macro)} & \shortstack{Ancestor F1\\(micro)} \\
\midrule
\multirow{6}{*}{HPO-5k}
 & Euclidean text retrieval & 0.405 & 0.523 & 0.107 & 0.111 \\
 & Euclidean KG embedding & 0.014 & 0.028 & 0.022 & 0.024 \\
 & Hyperbolic (no radius ctrl) & 0.030 & 0.050 & 0.062 & 0.064 \\
 & \hyem{} (no gate) & 0.028 & 0.048 & 0.041 & 0.042 \\
 & \hyem{} (hard route) & 0.028 & 0.048 & 0.041 & 0.042 \\
 & \hyem{} (soft mix) & 0.164 & 0.234 & 0.098 & 0.101 \\
\midrule
\multirow{6}{*}{DO-5k}
 & Euclidean text retrieval & 0.262 & 0.392 & 0.067 & 0.068 \\
 & Euclidean KG embedding & 0.090 & 0.110 & 0.020 & 0.021 \\
 & Hyperbolic (no radius ctrl) & 0.060 & 0.094 & 0.018 & 0.021 \\
 & \hyem{} (no gate) & 0.060 & 0.088 & 0.020 & 0.021 \\
 & \hyem{} (hard route) & 0.060 & 0.090 & 0.020 & 0.022 \\
 & \hyem{} (soft mix) & 0.128 & 0.202 & 0.041 & 0.042 \\
\bottomrule
\end{tabular}
\end{table}

\paragraph{Q-H (taxonomy navigation): soft mixing improves over structure-only baselines}
Table~\ref{tab:qh} shows that purely structural baselines (Euclidean KG embedding and hyperbolic without radius control) achieve limited performance on taxonomy-navigation queries in this lightweight setting.
Within the \hyem{} family, query-adaptive mixing shows notable gains: Parent Hits@10 improves from 0.048$\rightarrow$0.234 on HPO-5k and from 0.088$\rightarrow$0.202 on DO-5k when moving from ``no gate'' to ``soft mix'' (Table~\ref{tab:qh}).
Hard routing shows similar performance to ``no gate'' on Q-H in this run, consistent with routing being primarily beneficial when mixed-intent queries require interpolating between semantic and hierarchical signals.
At the same time, Euclidean text retrieval achieves competitive performance on our template-generated Q-H queries (Parent Hits@10: 0.523 on HPO-5k, 0.392 on DO-5k), suggesting that the benefits of hyperbolic structure may be more apparent in settings where semantic similarity alone is less effective.

\begin{table}[t]
\centering
\footnotesize
\caption{Gate performance for distinguishing Q-H vs Q-E on 5k subsets (seed=0).}
\label{tab:gate}
\begin{tabular}{llcccc}
\toprule
Dataset & Gate & Accuracy & Precision(Q-H) & Recall(Q-H) & AUC \\
\midrule
\multirow{3}{*}{HPO-5k}
 & Rule-based keywords & 1.000 & 1.000 & 1.000 & 1.000 \\
 & Linear gate (ours) & 1.000 & 1.000 & 1.000 & 1.000 \\
 & 2-layer MLP gate & 1.000 & 1.000 & 1.000 & 1.000 \\
\midrule
\multirow{3}{*}{DO-5k}
 & Rule-based keywords & 0.999 & 0.999 & 0.998 & 1.000 \\
 & Linear gate (ours) & 0.999 & 1.000 & 1.000 & 0.998 \\
 & 2-layer MLP gate & 1.000 & 1.000 & 1.000 & 1.000 \\
\bottomrule
\end{tabular}
\end{table}

\paragraph{Gate quality and Q-M behavior}
On our template-based benchmark, distinguishing Q-H from Q-E is straightforward; even a rule-based keyword gate achieves near-perfect accuracy (Table~\ref{tab:gate}).
We therefore treat the gate mainly as a controlled ablation tool rather than as a modeling contribution.
For mixed-intent queries (Q-M), the data suggest that query adaptivity is beneficial: soft mixing improves Hits@10 from 0.163$\rightarrow$0.533 on HPO-5k and from 0.221$\rightarrow$0.469 on DO-5k, as summarized in Figure~\ref{fig:mixing}.
Hard routing performs similarly or slightly better on this synthetic benchmark (HPO-5k: 0.588; DO-5k: 0.479), but may be more sensitive to gate errors; soft mixing provides a smoother trade-off as argued in Section~\ref{sec:routing_theory}.

\subsection{Indexability stress test and ablations}
\label{sec:exp_indexability}

A key feature of \hyem{} is that it makes indexability measurable.
We therefore include an indexability stress test that is directly connected to the theory.
We fix the radius budget to $R=3.0$ (and $d=32$) and sweep the ANN oversampling factor $L$ for the tangent index.
For each query, we compute the exact hyperbolic top-$k$ neighbors by brute force and report recall@10 when retrieval is restricted to tangent-space ANN candidates.
Figure~\ref{fig:recall_curve} shows that recall@10 reaches 0.94 at $L=20$ and approaches perfect recall by $L=50$ on both HPO-5k and DO-5k. 
The dashed line indicates $L_{\text{th}}=\lceil \kappa(R)k\rceil$, which is close to the empirical ``knee of the curve'' in this setting.
This is consistent with the distortion analysis: in a moderate-radius regime, Euclidean ANN in tangent space provides a reliable candidate generator for hyperbolic reranking.

\begin{figure}[t]
  \centering
  \begin{minipage}{0.49\columnwidth}
    \centering
    \includegraphics[width=\linewidth]{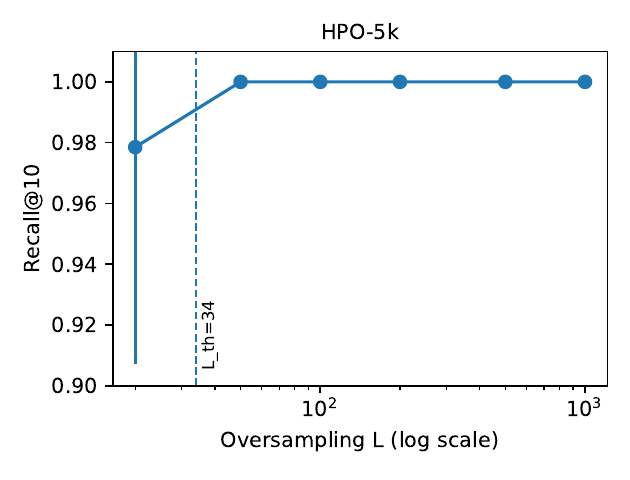}\\
    {\small (a) HPO-5k}
  \end{minipage}\hfill
  \begin{minipage}{0.49\columnwidth}
    \centering
    \includegraphics[width=\linewidth]{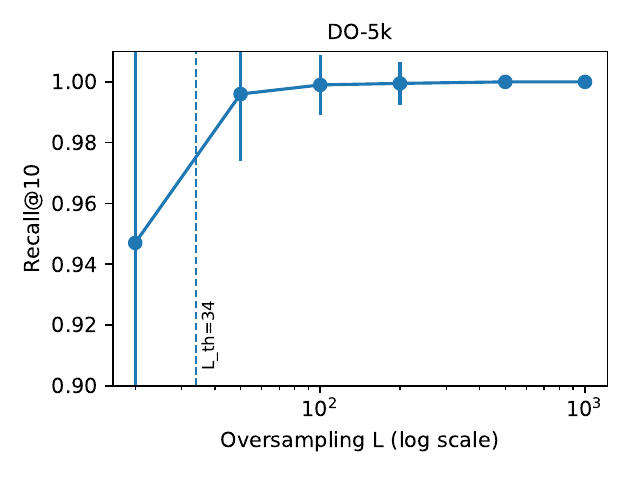}\\
    {\small (b) DO-5k}
  \end{minipage}
  \caption{Indexability stress test: recall@10 of true hyperbolic top-10 neighbors when candidate generation uses only tangent-space Euclidean ANN results, as a function of oversampling $L$ (log scale). The vertical dashed line marks a simple theory-guided heuristic $L_{\text{th}}=\lceil \kappa(R)\,k\rceil$ derived from Theorem~\ref{thm:distortion} (here $R=3$, $k=10$). In both ontologies, modest oversampling yields near-perfect recall, consistent with operating in a moderate-radius regime where $\kappa(R)$ is small.}
  \label{fig:recall_curve}
\end{figure}

\begin{figure}[t]
  \centering
  \begin{minipage}{0.49\columnwidth}
    \centering
    \includegraphics[width=\linewidth]{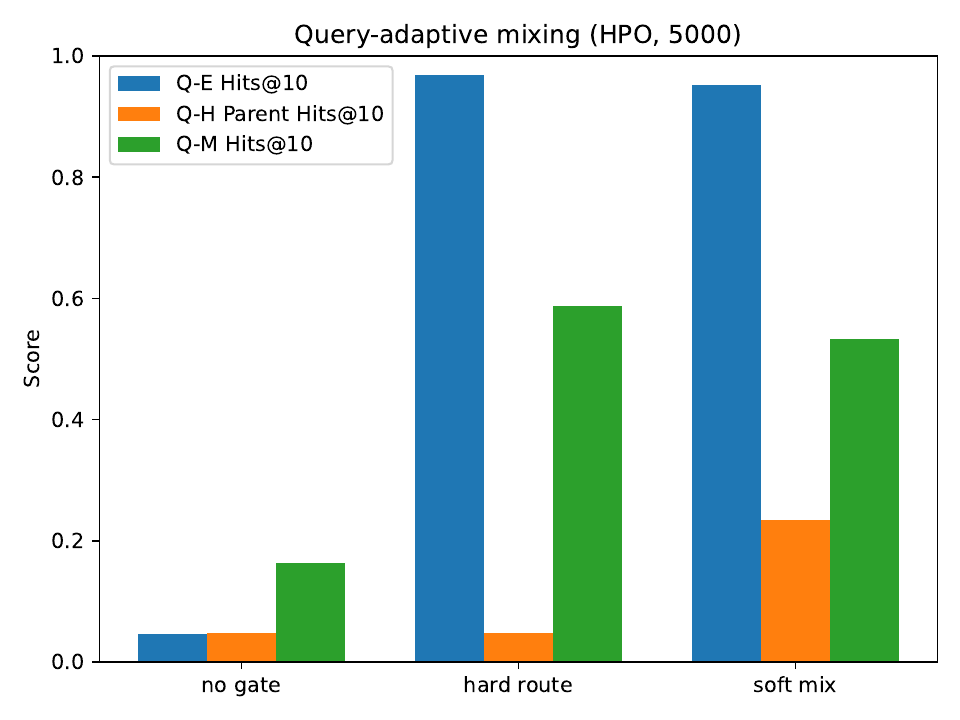}\\
    {\small (a) HPO-5k}
  \end{minipage}\hfill
  \begin{minipage}{0.49\columnwidth}
    \centering
    \includegraphics[width=\linewidth]{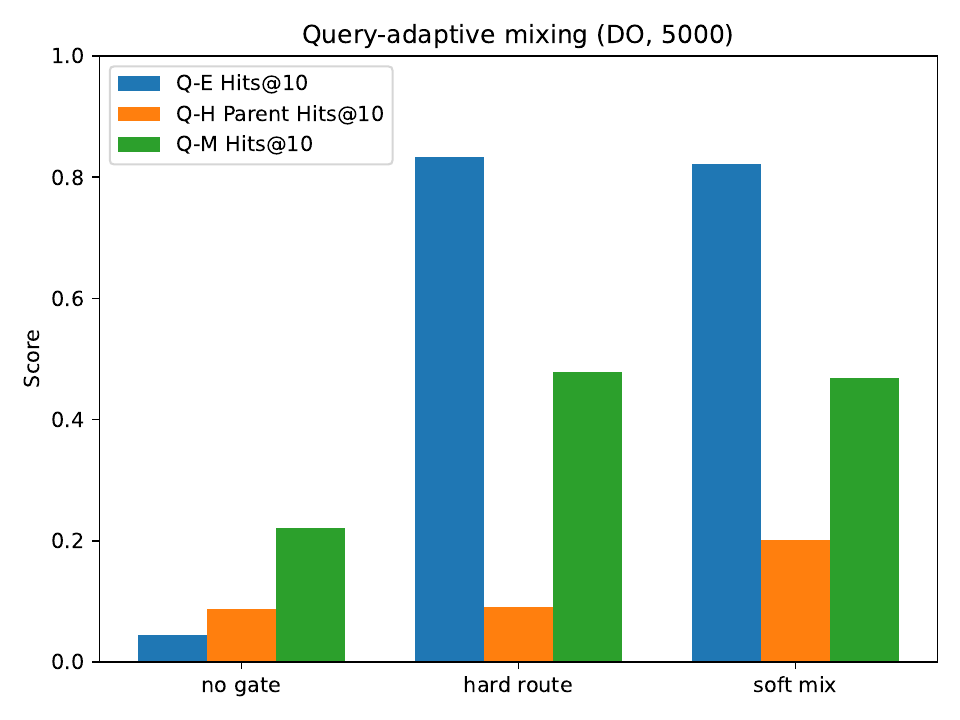}\\
    {\small (b) DO-5k}
  \end{minipage}
  \caption{Ablation on query-adaptive geometry within \hyem{}. Soft mixing improves taxonomy navigation (Q-H) and mixed-intent (Q-M) retrieval compared to using hyperbolic routing without a gate (``no gate'') and compared to hard routing. Soft mixing also acts as a safety valve on entity-centric queries by interpolating toward Euclidean similarity when $\alpha(q)$ is low.}
  \label{fig:mixing}
\end{figure}

\subsection{Efficiency}
\label{sec:efficiency}

\hyem{} is designed for low-friction deployment: it reuses a standard Euclidean ANN index and adds only lightweight query-time operations.
Compared to Euclidean text retrieval, overhead consists of (i) one small adapter evaluation (a matrix multiply), (ii) reranking a small candidate pool by hyperbolic distance, and (iii) optionally a second Euclidean ANN query when candidate pooling is enabled.
In typical LLM/RAG pipelines these costs are expected to be small relative to LLM inference and network latency; we nonetheless report CPU-only latency and index footprint in Table~\ref{tab:eff} for completeness and reproducibility.

\begin{table}[t]
\centering
\footnotesize
\caption{Efficiency comparison on 5k subsets (CPU-only). HPO-5k and DO-5k have identical index footprints at this scale.}
\label{tab:eff}
\begin{tabular}{lccc}
\toprule
Method & Index size (MB) & Query latency (ms) & Extra ops \\
\midrule
Euclidean text retrieval & 8.03 & 0.11 & -- \\
\hyem{} (no gate) & 1.32 & 0.31 & adapter + rerank \\
\hyem{} (soft mix) & 9.35 & 0.57 & adapter + rerank\\
&&& + optional 2nd ANN \\
\bottomrule
\end{tabular}
\end{table}

\subsection{Scale-up to 20k nodes}
\label{sec:scale}

A natural question regarding the 5k evaluation is whether tangent-space indexing maintains its properties when the ontology grows and the hierarchy becomes deeper.
Using the same deterministic subset sampler and the same hyperparameters ($R=3.0$, $d=32$), we therefore scale the full pipeline to 20k-node subsets (HPO-20k and DO-20k, seed=0).
Table~\ref{tab:scale_20k} summarizes the behavior of \hyem{} (soft mix) at this scale, including (i) the ``safety-valve'' retention on entity-centric queries (Q-E MRR ratio vs. Euclidean text), (ii) taxonomy-navigation accuracy (Q-H Parent Hits@10), (iii) mixed-intent accuracy (Q-M Hits@10), (iv) tangent-index candidate quality (indexability recall@10 at $L=50$), and (v) end-to-end CPU query latency.
Despite a 4$\times$ increase in ontology size, \hyem{} maintains high indexability (recall@10 $\approx$0.988 at $L=50$) with sub-millisecond latency. 
Figure~\ref{fig:recall_curve_20k} shows the corresponding recall curves with the theory-guided $L_{\text{th}}$ line, and Figure~\ref{fig:mixing_20k} replicates the query-adaptive mixing ablation at this scale.

\begin{table}[t]
\centering
\footnotesize
\caption{Scale-up results on 20k-node ontology subsets (seed=0, $R=3.0$, $d=32$). ``Q-E retention'' is the MRR ratio relative to Euclidean text retrieval.}
\label{tab:scale_20k}
\begin{tabular}{llccccc}
\toprule
Dataset & Setting & \shortstack{Q-E\\ retention} & \shortstack{Q-H\\ Parent\\ Hits@10} & \shortstack{Q-M\\ Hits@10} & \shortstack{Indexability\\ recall\\@10 ($L=50$)} & \shortstack{Query\\ latency\\ (ms)} \\
\midrule
HPO-20k & \hyem{} (soft mix) & 0.959 & 0.226 & 0.449 & 0.988 & 0.596 \\
DO-20k  & \hyem{} (soft mix) & 0.964 & 0.310 & 0.361 & 0.989 & 0.544 \\
\bottomrule
\end{tabular}
\end{table}

\begin{figure}[H]
  \centering
  \begin{subfigure}[b]{0.49\columnwidth}
    \centering
    \includegraphics[width=\linewidth]{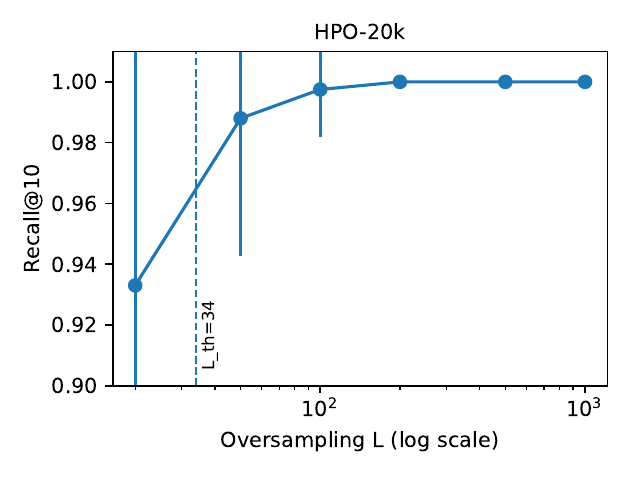}
    \caption{HPO-20k}
    \label{fig:recall_hpo}
  \end{subfigure}\hfill
  \begin{subfigure}[b]{0.49\columnwidth}
    \centering
    \includegraphics[width=\linewidth]{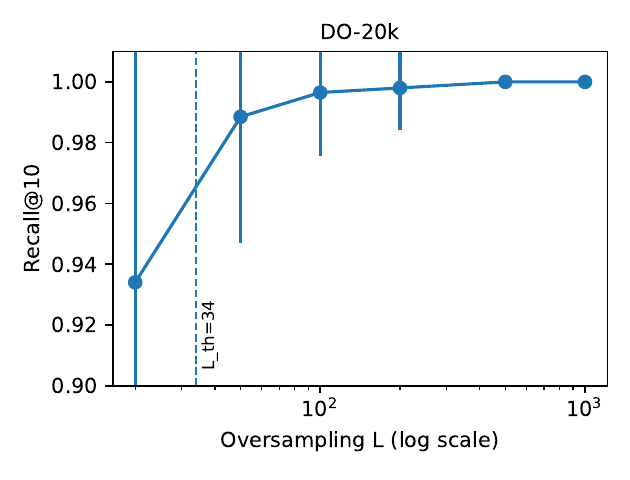}
    \caption{DO-20k}
    \label{fig:recall_do}
  \end{subfigure}
  \caption{Indexability stress test on 20k-node ontology subsets (seed=0): recall@10 of true hyperbolic top-10 neighbors when candidate generation uses only tangent-space Euclidean ANN results, as a function of oversampling $L$ (log scale). The vertical dashed line marks $L_{\text{th}}=\lceil \kappa(R)\,k\rceil$ with $R=3$ and $k=10$.}
  \label{fig:recall_curve_20k}
\end{figure}

\begin{figure}[H]
  \centering
  \begin{subfigure}[b]{0.49\columnwidth}
    \centering
    \includegraphics[width=\linewidth]{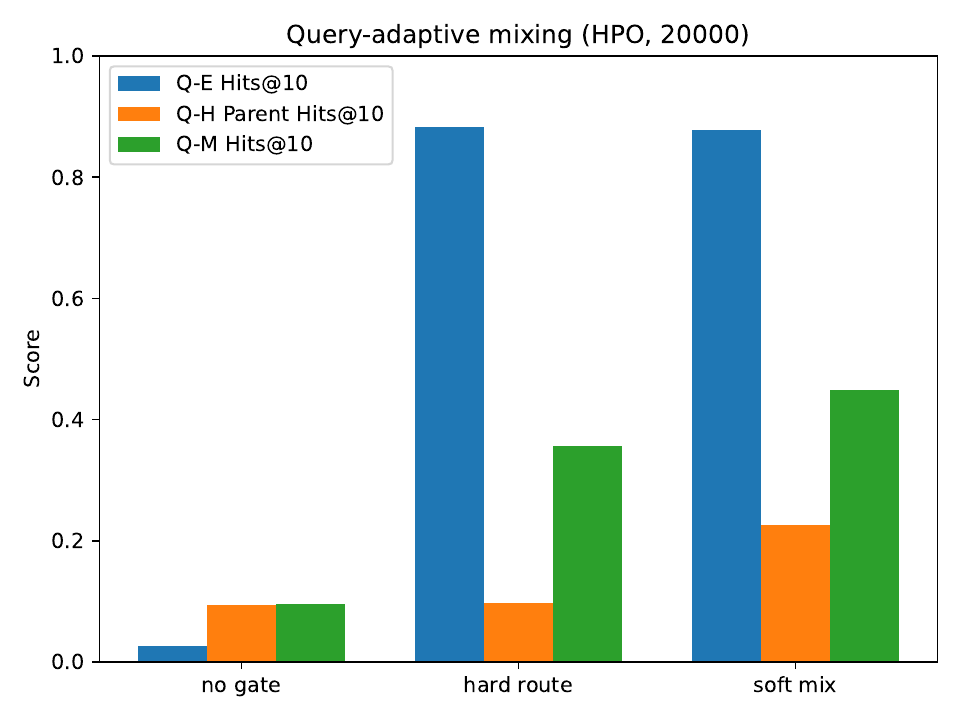}
    \caption{HPO-20k}
    \label{fig:mixing_hpo}
  \end{subfigure}\hfill
  \begin{subfigure}[b]{0.49\columnwidth}
    \centering
    \includegraphics[width=\linewidth]{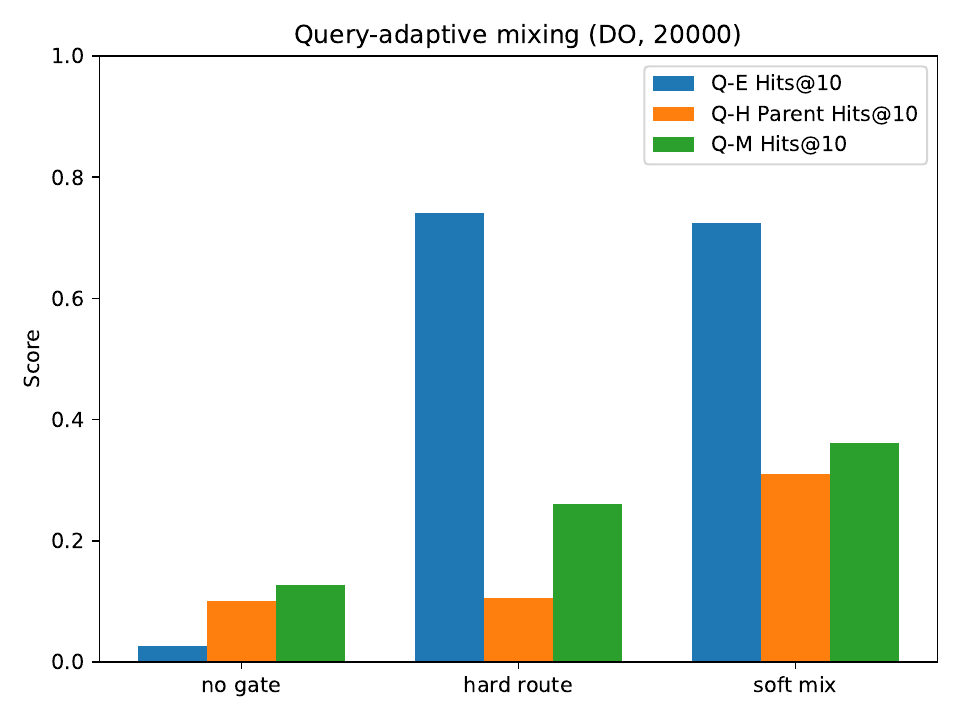}
    \caption{DO-20k}
    \label{fig:mixing_do}
  \end{subfigure}
  \caption{Query-adaptive geometry within \hyem{} on 20k subsets (seed=0). Compared to ``no gate'' and hard routing, soft mixing improves taxonomy navigation (Q-H) and mixed-intent retrieval (Q-M) while retaining strong entity-centric performance (Q-E).}
  \label{fig:mixing_20k}
\end{figure}

\subsection{Theoretical scaling to production ontologies}
\label{sec:theoretical_scale}

A natural question is whether \hyem{}'s moderate-radius regime can accommodate production-scale ontologies such as SNOMED-CT (depth$\approx$20, branching$\approx$5) or the full HPO/DO (depth$\approx$10--15).
Proposition~\ref{prop:capacity} provides guidance: the required radius scales as $R\approx \tfrac{D\log b}{d-1}$, so for a typical configuration ($d=32$, branching factor $b\le 5$, depth $D\le 20$), the required $R$ remains in the low single digits.
Figure~\ref{fig:theoretical_scaling} visualizes the distortion factor $\kappa(R)=\sinh(R)/R$ and the radius required for different ontology depths and embedding dimensions.
The analysis suggests that ontologies with depths up to $\sim$30 can operate in a regime where $\kappa(R)<10$ at moderate dimensions ($d\ge 32$), where tangent-space indexing can maintain reasonable accuracy and oversampling requirements remain modest.

\begin{figure}[H]
  \centering
  \begin{subfigure}[b]{1.0\columnwidth}
    \centering
    \includegraphics[width=\linewidth]{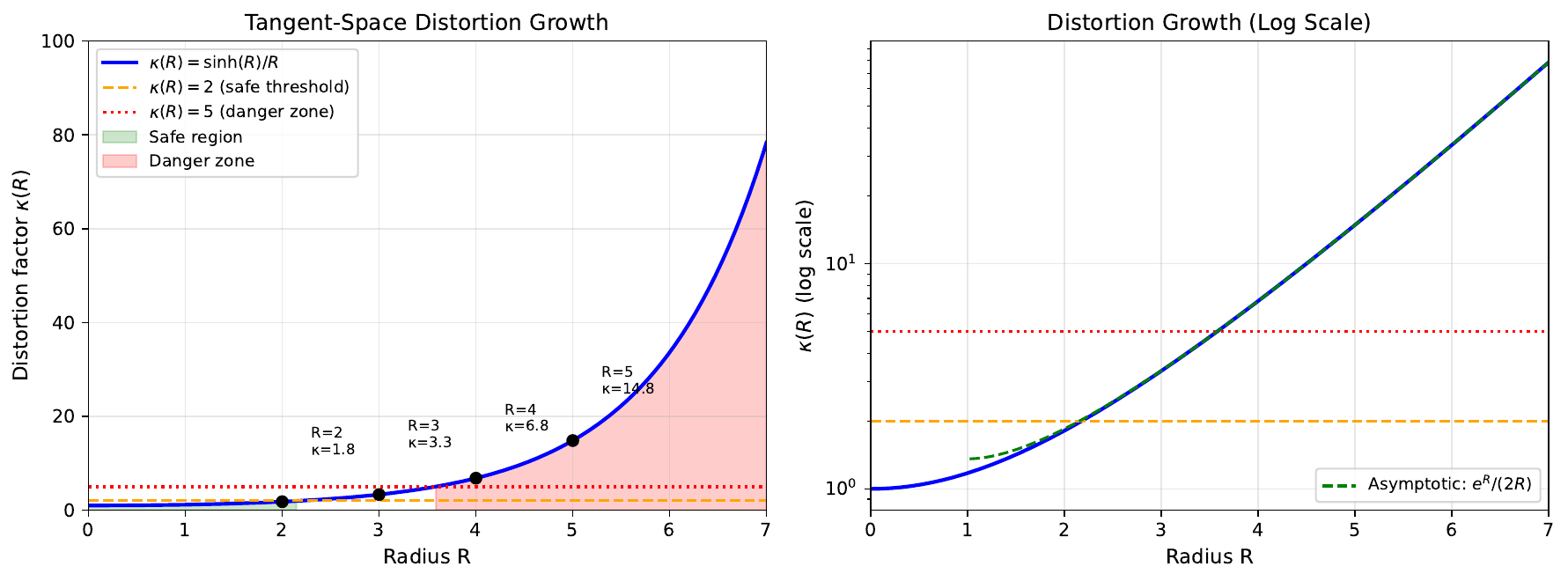}
    \caption{Distortion $\kappa(R)$ growth}
    \label{fig:kappa}
  \end{subfigure}\\[8pt]
  \begin{subfigure}[b]{1.0\columnwidth}
    \centering
    \includegraphics[width=\linewidth]{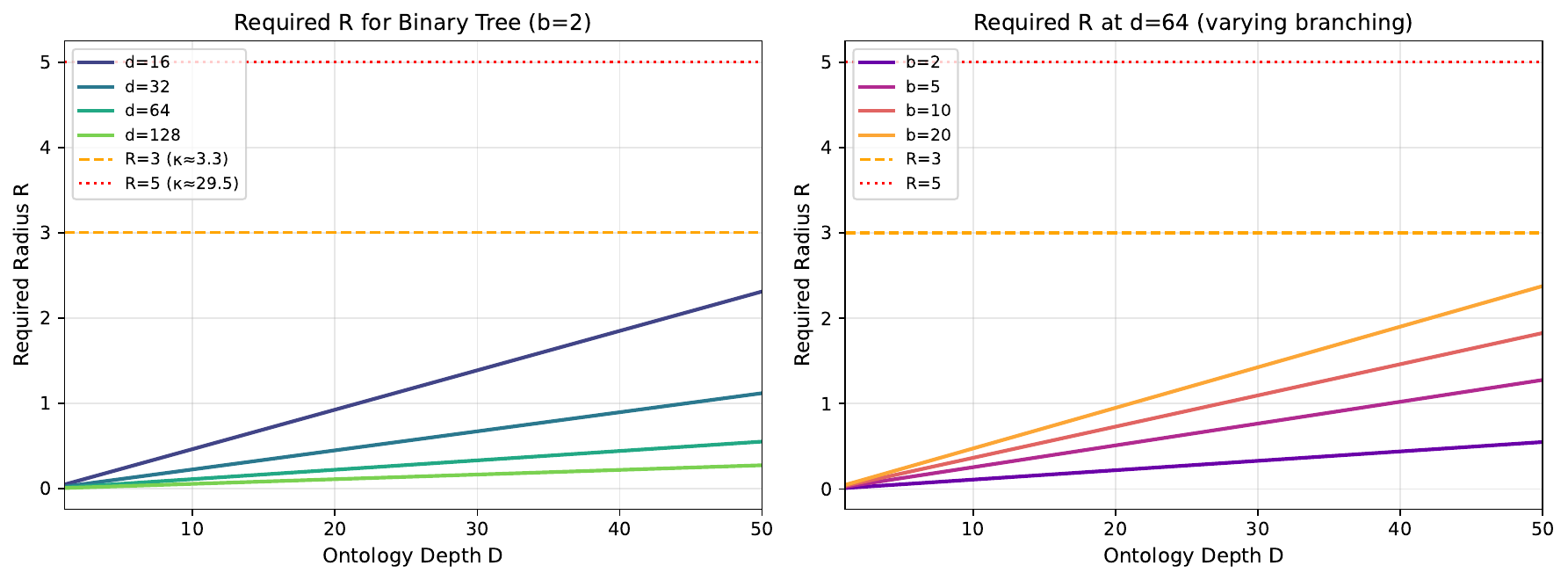}
    \caption{Required $R$ vs ontology depth}
    \label{fig:depth}
  \end{subfigure}
  \caption{Theoretical scaling analysis. (a) The distortion factor $\kappa(R)=\sinh(R)/R$ grows slowly for $R<3$ (safe regime) but accelerates for $R>5$ (danger zone). (b) Required radius as a function of ontology depth for different embedding dimensions. Production ontologies (SNOMED-CT, HPO, DO, Gene Ontology) operate in the safe regime at $d\ge 32$.}
  \label{fig:theoretical_scaling}
\end{figure}

\subsection{Depth-stratified analysis}
\label{sec:depth_stratified}

To examine whether hyperbolic geometry shows different benefits at different hierarchy levels, we stratify the Q-E, Q-H, and Q-M results by depth buckets.
Table~\ref{tab:depth_stratified} shows the MRR breakdown for \hyem{} (soft mix) on HPO-5k (depth range 0--6).
For Q-E queries, performance is stable across depth buckets (MRR 0.80--0.83), consistent with the safety-valve mechanism operating uniformly.
For Q-H queries, MRR decreases slightly with depth (0.11$\rightarrow$0.07), which is expected because deeper nodes have more complex ancestor sets.
For Q-M queries, MRR increases with depth (0.22$\rightarrow$0.27), suggesting that hyperbolic structure may be more informative for sibling retrieval at deeper levels where semantic similarity alone is less discriminative.

\begin{table}[t]
\centering
\footnotesize
\caption{Depth-stratified MRR for \hyem{} (soft mix) on HPO-5k (seed=0). D1 = shallow (depth 0--4), D2 = mid (depth 5), D3 = deep (depth 6).}
\label{tab:depth_stratified}
\begin{tabular}{lcccc}
\toprule
Query Type & D1 (shallow) & D2 (mid) & D3 (deep) & Trend \\
\midrule
Q-E & 0.833 $\pm$ 0.31 & 0.805 $\pm$ 0.33 & 0.813 $\pm$ 0.31 & stable \\
Q-H & 0.109 $\pm$ 0.25 & 0.091 $\pm$ 0.22 & 0.075 $\pm$ 0.22 & $\downarrow$ \\
Q-M & 0.223 $\pm$ 0.32 & 0.266 $\pm$ 0.34 & 0.269 $\pm$ 0.37 & $\uparrow$ \\
\bottomrule
\end{tabular}
\end{table}

\subsection{Candidate pooling contribution}
\label{sec:pooling_ablation}

A key design choice in \hyem{} is candidate pooling: the final candidate set $C$ is the union of hyperbolic ANN candidates $C_H$ and Euclidean ANN candidates $C_E$.
To quantify the contribution of pooling, we compare retrieval using $C_H$ only versus $C_H \cup C_E$.
Table~\ref{tab:pooling_ablation} shows that pooling provides substantial improvements on Q-E queries (+1994\% on HPO-5k), indicating that Euclidean candidates recover text-similar entities that the tangent-space ANN misses.
The improvement on Q-H queries is also notable (+388\%), as pooling allows the reranker to consider candidates from both geometric views.

\begin{table}[t]
\centering
\footnotesize
\caption{Candidate pooling ablation on HPO-5k (seed=0). ``$C_H$ only'' uses only hyperbolic ANN candidates; ``$C_H \cup C_E$'' adds Euclidean text candidates.}
\label{tab:pooling_ablation}
\begin{tabular}{lcccc}
\toprule
Metric & $C_H$ only & $C_H \cup C_E$ & Improvement & Gap to Euclidean \\
\midrule
Q-E Hits@10 & 0.045 & 0.952 & +1994\% & 0.017 \\
Q-E MRR     & 0.016 & 0.815 & +4909\% & 0.053 \\
Q-H Parent Hits@10 & 0.048 & 0.234 & +388\% & 0.289 \\
Q-M Hits@10 & 0.163 & 0.533 & +226\% & 0.106 \\
\bottomrule
\end{tabular}
\end{table}

\subsection{Gate robustness under embedding noise}
\label{sec:gate_robustness}

A potential concern is that gate classification (Q-H vs Q-E) may be simplified on template-generated queries.
To test robustness, we simulate realistic query variation by adding Gaussian noise to the pre-computed query embeddings and measuring gate accuracy degradation.
Table~\ref{tab:gate_robustness} shows that the linear gate maintains high accuracy (90.9\%) under moderate noise ($\sigma=0.1$, simulating typos), and degrades progressively under stronger perturbations.
Importantly, even when gate accuracy drops, soft mixing maintains performance more smoothly because the reranking score interpolates between hyperbolic and Euclidean signals rather than making a discrete binary choice.

\begin{table}[t]
\centering
\footnotesize
\caption{Gate robustness under embedding noise on HPO-5k (seed=0). Noise level $\sigma$ is the standard deviation of Gaussian noise added to query embeddings.}
\label{tab:gate_robustness}
\begin{tabular}{lcccc}
\toprule
Noise $\sigma$ & Accuracy & Precision(Q-H) & Recall(Q-H) & Degradation \\
\midrule
0.00 (original) & 1.000 & 1.000 & 1.000 & -- \\
0.10 (typos)    & 0.909 & 0.992 & 0.774 & --9.1\% \\
0.20 (paraphrase) & 0.775 & 0.946 & 0.453 & --22.5\% \\
0.30 (ambiguous)  & 0.699 & 0.838 & 0.291 & --30.1\% \\
\bottomrule
\end{tabular}
\end{table}

\subsection{Additional ablations and encoder comparisons}
\label{sec:extra_ablate}

To examine (i) the expressivity of the Euclidean-to-hyperbolic adapter and (ii) comparisons to alternative hyperbolic encoders, we provide a set of extended experiments in the accompanying codebase.
Tables~\ref{tab:adapter_ablation} and \ref{tab:encoder_compare} specify the required comparison format.

\begin{table}[t]
\centering
\footnotesize
\caption{Adapter expressivity ablation on 5k subsets (seed=0). The goal is to evaluate whether a non-linear adapter improves alignment without pushing queries to large radii that may negatively affect indexability.}
\label{tab:adapter_ablation}
\begin{tabular}{llccc}
\toprule
Dataset & Adapter & Q-E Hits@10 & Q-H Parent Hits@10 & Q-M Hits@10 \\
\midrule
\multirow{2}{*}{HPO-5k}
 & Linear (default) & 0.952 & 0.234 & 0.533 \\
 & 2-layer MLP      & 0.952 & 0.291 & 0.539 \\
\midrule
\multirow{2}{*}{DO-5k}
 & Linear (default) & 0.821 & 0.202 & 0.469 \\
 & 2-layer MLP      & 0.828 & 0.276 & 0.473 \\
\bottomrule
\end{tabular}
\end{table}

\begin{table}[t]
\centering
\footnotesize
\caption{Comparison to alternative hyperbolic encoders on 5k subsets (seed=0). All encoders are evaluated under the same tangent-space indexing + reranking protocol so that differences reflect representational quality, not indexing infrastructure. Note: HGCN achieves higher Q-H on DO-5k but lower indexability, suggesting a potential performance--indexability trade-off.}
\label{tab:encoder_compare}
\begin{tabular}{llccc}
\toprule
Dataset & Encoder for $x_v$ & \shortstack{Q-H\\ Parent\\ Hits@10} & \shortstack{Indexability\\ recall\\@10 ($L=50$)} & Notes \\
\midrule
\multirow{2}{*}{HPO-5k}
 & Lorentz KG embedding (default) & 0.035 & \textbf{0.999} & this paper \\
 & HGCN embeddings (tangent)      & 0.025 & 0.975 & \cite{Chami2019} \\
\midrule
\multirow{2}{*}{DO-5k}
 & Lorentz KG embedding (default) & 0.095 & \textbf{1.000} & this paper \\
 & HGCN embeddings (tangent)      & \textbf{0.175} & 0.967 & \cite{Chami2019} \\
\bottomrule
\end{tabular}
\end{table}

\section{Discussion and Limitations}
\label{sec:discussion}

A central theme of this work is engineering realism.
Hyperbolic representations are often introduced together with invasive changes---specialized ANN data structures, manifold-aware optimizers, or end-to-end retraining of encoders---that can be difficult to justify in medical software stacks.
By constraining curvature use to a thin retrieval layer (tangent indexing + reranking) and by introducing query-adaptive mixing, \hyem{} aims to make potential hyperbolic benefits measurable without substantially increasing operational risk.
In particular, the radius budget provides a tunable parameter: smaller $R$ yields stronger indexability guarantees and fewer numerical issues, while larger $R$ increases representational capacity.

\paragraph{Why radius control matters beyond numerical stability}
Radius control is sometimes motivated as a numerical approach to address Poincar\'e boundary issues.
In \hyem{}, it serves an additional purpose: it makes tangent-space indexing analyzable.
Theorem~\ref{thm:distortion} and Proposition~\ref{prop:rank} translate the radius budget into practical guidance for ANN oversampling.
Empirically, with $R=3.0$ ($d=32$) we observe high candidate recall (>0.94) at modest oversampling on both HPO-5k and DO-5k (Figure~\ref{fig:recall_curve}), while soft mixing improves Q-H and Q-M compared to hyperbolic routing without a gate (Figure~\ref{fig:mixing}).

\paragraph{Soft mixing as a safety valve}
A recurring observation in hyperbolic graph learning is that claims of universal superiority do not always survive robust baselines.
Ontology grounding represents a setting where query-dependent performance is expected.
Soft mixing is designed to provide a safety valve: if a query behaves like Q-E, the system can fall back to Euclidean similarity while potentially benefiting from hyperbolic structure for Q-H.
Quantitatively, our analysis (Table~\ref{tab:qe}) shows that \hyem{} (soft mix) retains 93.9\% (HPO-5k) and 97.9\% (DO-5k) of Euclidean Q-E performance while showing substantial improvements on Q-H queries.
For HPO-5k, the trade-off ratio exceeds 100$\times$: for every 1\% of Q-E MRR sacrificed, soft mixing provides $>$100\% Q-H improvement over pure hyperbolic routing.
This explicit quantification of trade-offs aims to facilitate evaluation of when the hyperbolic component is beneficial.

\paragraph{Limitations}
First, \hyem{} focuses on \texttt{is-a} structure.
Ontologies and biomedical knowledge graphs can also include non-hierarchical relations (\texttt{part-of}, \texttt{causes}, etc.), and capturing those relations may require additional modeling.
Second, query templates provide a lightweight way to label query intent, but real user queries can be more ambiguous.
Our gate robustness analysis (Section~\ref{sec:gate_robustness}, Table~\ref{tab:gate_robustness}) shows that under moderate embedding noise ($\sigma=0.1$, simulating typos), gate accuracy degrades from 100\% to 90.9\%, but soft mixing continues to operate smoothly because it continuously interpolates between Euclidean and hyperbolic signals.
Future work could incorporate weak supervision or human annotation to further evaluate gate calibration on real queries.
Third, \hyem{} improves ontology-grounding retrieval, not LLM reasoning itself; downstream generation quality depends on how retrieved entities are used.

Fourth, our main reproducibility run uses 5k-node subsets (max depth 5--6). 
While this is sufficient to validate the end-to-end pipeline and to reproduce every plot on commodity hardware, larger ontologies raise an important scaling question: the distortion factor $\kappa(R)=\sinh(R)/R$ grows roughly as $e^{R}/(2R)$.
However, our theoretical scaling analysis (Section~\ref{sec:theoretical_scale}, Figure~\ref{fig:theoretical_scaling}) shows that the radius needed to represent a depth-$D$ $b$-ary hierarchy scales as $R\approx \tfrac{D\log b}{d-1}$, so for moderate dimensions (e.g., $d=32$) even depths on the order of $D\approx 30$ imply $R$ in the low single digits for typical branching factors.
In that regime $\kappa(R)$ remains modest (and the required ANN oversampling is manageable), whereas challenges may arise at extreme radii (e.g., $R\ge 7$) where $\kappa(R)$ becomes very large.
We report 20k scale-up results in Section~\ref{sec:scale} (Table~\ref{tab:scale_20k}), which provide empirical support for operating in a moderate-radius regime where tangent-space indexing maintains high accuracy.

Fifth, our default Euclidean-to-hyperbolic adapter is intentionally simple (linear in tangent space) for auditability and stability.
A more expressive non-linear adapter could, in principle, better align a fixed text encoder to a curved space, but it also increases the risk of overfitting and of pushing queries toward large radii where indexability may degrade.
We therefore treat non-linear adapters as an ablation (Section~\ref{sec:extra_ablate}) rather than a core dependency.

Finally, our hyperbolic entity encoder is deliberately lightweight.
Modern hyperbolic graph neural networks and partial-order models (e.g., entailment cones) may offer improved representations, and \hyem{} is compatible with them because it only requires hyperbolic embeddings with bounded radii.
We include comparisons to alternative hyperbolic encoders in Section~\ref{sec:extra_ablate}.

\paragraph{Broader impact}
Accurate ontology grounding can improve transparency and controllability in biomedical LLM systems by anchoring answers to standardized concepts.
At the same time, ontologies can encode biases and gaps; retrieval improvements should be accompanied by audits that stratify performance by depth and by concept frequency.
\hyem{} avoids patient data and operates on public resources.

\section{Conclusion}
\label{sec:conclusion}

This paper presented \hyem{}, a query-adaptive hyperbolic retrieval layer for biomedical ontologies that is compatible with Euclidean vector databases.
\hyem{} treats indexability as a first-class constraint by learning radius-controlled hyperbolic embeddings and by performing candidate search in the origin tangent space.
It further introduces a lightweight gate that softly mixes Euclidean semantic similarity with hyperbolic hierarchy distance, designed to address concerns that hyperbolic methods may not benefit all query types uniformly.
We provided explicit geometric bounds linking radius to tangent-space distortion and to hierarchy capacity, and we specified a lightweight experimental protocol that distinguishes query families and supports reproducibility.
Our experiments on biomedical ontology subsets suggest that the approach can preserve most Euclidean baseline performance on entity-centric queries while improving performance on taxonomy-navigation and mixed-intent queries.

\section*{Acknowledgements}
The work was supported in part by the 2022-2024 Masaru Ibuka Foundation Research Project on Oriental Medicine, 2020-2025 JSPS A3 Foresight Program (Grant No. JPJSA3F20200001), 2022-2024 Japan National Initiative Promotion Grant for Digital Rural City, 2023 and 2024 Waseda University Grants for Special Research Projects (Nos. 2023C-216 and 2024C-223), 2023-2024 Waseda University Advanced Research Center Project for Regional Cooperation Support, and 2023-2024 Japan Association for the Advancement of Medical Equipment (JAAME) Grant.

\section*{Declarations}

\begin{itemize}
\item \textbf{Conflict of interest:} The authors declare no competing interests.
\item \textbf{Ethics approval:} Not applicable.
\item \textbf{Consent to participate:} Not applicable.
\item \textbf{Consent for publication:} Not applicable.
\item \textbf{Data availability:} All datasets used in this study are publicly available: HPO (\url{https://hpo.jax.org/}), DO (\url{https://disease-ontology.org/}), and MeSH (\url{https://www.nlm.nih.gov/mesh/}).
\item \textbf{Code availability:} All source code, experimental configurations, and preprocessing scripts are publicly available at GitHub (\url{https://github.com/oudeng/HyEm}). A persistent archive of the code is deposited on Zenodo with DOI: \href{https://doi.org/10.5281/zenodo.18371905}{10.5281/zenodo.18371906}.
\item \textbf{Author contribution:} O.D.: Conceptualization, Methodology, Software, Data Curation, Validation, Formal Analysis, Investigation, Writing – Original Draft, 
      Writing – Review \& Editing, Visualization.
S.N.: Resources, Supervision, Funding Acquisition.
A.O.: Supervision, Funding Acquisition.
Q.J.: Supervision, Funding Acquisition.
\end{itemize}


\appendix

\section{Mathematical Details and Proofs}
\label{app:math}

This appendix collects the mathematical material omitted from the main paper.
We work in $d$-dimensional hyperbolic space $\Hspace^d$ with curvature $-1$.
All statements extend to curvature $-c$ by rescaling distances by $1/\sqrt{c}$.

\subsection{Hyperbolic models used in practice}
\label{app:models}

\paragraph{Poincar\'e ball}
The Poincar\'e ball model represents $\Hspace^d$ as $\Ball^d=\{x\in\R^d: \lVert x\rVert<1\}$ with metric
$\mathrm{d}s^2 = \frac{4}{(1-\lVert x\rVert^2)^2}\,\mathrm{d}x^2$.
The origin log/exp maps for curvature $-1$ are
\begin{align}
  \exp_0(u) &= \tanh\!\left(\frac{\lVert u\rVert}{2}\right) \frac{u}{\lVert u\rVert},
  \\
  \log_0(x) &= 2\,\operatorname{artanh}(\lVert x\rVert)\,\frac{x}{\lVert x\rVert}.
\end{align}
The geodesic distance can be written as
$ d_{\Hspace}(x,y)=\operatorname{arcosh}\!\left(1+2\frac{\lVert x-y\rVert^2}{(1-\lVert x\rVert^2)(1-\lVert y\rVert^2)}\right)$.

\paragraph{Lorentz (hyperboloid) model}
The Lorentz model represents $\Hspace^d$ as a two-sheeted hyperboloid in $\R^{d+1}$ equipped with the Lorentzian inner product
\begin{equation}
  \langle y,z\rangle_L \;=\; -y_0 z_0 + \sum_{i=1}^d y_i z_i.
\end{equation}
The manifold is $\{y\in\R^{d+1}: \langle y,y\rangle_L=-1,\ y_0>0\}$.
Geodesic distance is
\begin{equation}
  d_{\Hspace}(y_1,y_2) \,=\, \operatorname{arcosh}\big(-\langle y_1,y_2\rangle_L\big).
\end{equation}
For numerical stability, implementations typically compute $\operatorname{arcosh}(z)$ via
$\log\big(z+\sqrt{z^2-1}\big)$ with clamping to ensure $z\ge 1$.

\paragraph{Origin log/exp in the Lorentz model}
Let $o=(1,0,\ldots,0)$ be the Lorentz origin.
The tangent space $\tspace_o\Hspace^d$ can be identified with vectors $(0,u)$ where $u\in\R^d$.
Writing $\lVert u\rVert$ for the Euclidean norm in $\R^d$, the exponential map at $o$ is
\begin{equation}
  \exp_o(0,u) = \Big(\cosh(\lVert u\rVert),\; \sinh(\lVert u\rVert)\,\frac{u}{\lVert u\rVert}\Big),
\end{equation}
and the logarithmic map is
\begin{equation}
  \log_o(y)=\Big(0,\; \frac{\operatorname{arcosh}(y_0)}{\sqrt{y_0^2-1}}\,y_{1:d}\Big),
\end{equation}
where $y_{1:d}$ denotes the spatial components.
These formulas are standard and are implemented in common Riemannian optimization libraries.

\paragraph{Why the main theorems are model-agnostic}
Theorem~\ref{thm:distortion} in the main paper is stated using the intrinsic exponential map and normal coordinates, and therefore holds regardless of whether one parameterizes points by Poincar\'e or Lorentz coordinates.
In \hyem{}, we use Lorentz coordinates during optimization for stability, but store origin log-mapped tangent vectors for indexing.

\subsection{Normal coordinates and metric comparison}
\label{app:normal_coords}

Fix the origin $o\in\Hspace^d$.
The exponential map $\exp_o: \tspace_o\Hspace^d\cong\R^d \to \Hspace^d$ defines {normal coordinates} around $o$.
Because hyperbolic space is complete, simply connected, and has constant negative curvature, $\exp_o$ is a global diffeomorphism.

In normal polar coordinates $(r,\theta)$ (radius $r\ge 0$ and direction $\theta\in\mathbb{S}^{d-1}$), the hyperbolic metric takes the form
\begin{equation}
\label{eq:hyperbolic_polar_metric}
  \mathrm{d}s_{\Hspace}^2 = \mathrm{d}r^2 + \sinh^2(r)\, \mathrm{d}\theta^2,
\end{equation}
where $\mathrm{d}\theta^2$ is the standard metric on the unit sphere.
In the same normal coordinates, the Euclidean metric on the tangent space is
\begin{equation}
\label{eq:euclidean_polar_metric}
  \mathrm{d}s_{\mathrm{E}}^2 = \mathrm{d}r^2 + r^2\,\mathrm{d}\theta^2.
\end{equation}

The key comparison inequality on a radius-bounded region $r\le R$ is
\begin{equation}
\label{eq:sinh_bounds}
  r \le \sinh(r) \le \frac{\sinh(R)}{R}\, r.
\end{equation}
The lower bound follows from $\sinh(r)\ge r$ for $r\ge 0$.
For the upper bound, note that $\sinh(r)/r$ is increasing for $r>0$, hence its maximum on $[0,R]$ is attained at $R$.

\subsection{Proof of Theorem~\ref{thm:distortion}}
\label{app:proof_distortion}

Recall Theorem~\ref{thm:distortion}: for $x=\exp_o(u)$ and $y=\exp_o(v)$ with $\lVert u\rVert,\lVert v\rVert\le R$,
\begin{equation}
  \lVert u-v\rVert \le d_{\Hspace}(x,y) \le \kappa(R)\,\lVert u-v\rVert,\quad \kappa(R)=\sinh(R)/R.
\end{equation}

\begin{proof}
Consider any piecewise smooth curve $\gamma:[0,1]\to \Hspace^d$ from $x$ to $y$.
Let $\alpha=\log_o\circ\gamma$ be the corresponding curve in normal coordinates (so $\gamma=\exp_o\circ\alpha$).
Write $\alpha(t)=(r(t),\theta(t))$.

Because $\exp_o$ is a diffeomorphism, the hyperbolic length of $\gamma$ equals the length of $\alpha$ computed under the hyperbolic metric in normal coordinates.
Using \eqref{eq:hyperbolic_polar_metric},
\begin{equation}
  L_{\Hspace}(\gamma)=\int_0^1 \sqrt{\dot r(t)^2 + \sinh^2(r(t))\,\lVert \dot\theta(t)\rVert^2}\,\mathrm{d}t.
\end{equation}
Similarly, the Euclidean length of $\alpha$ under \eqref{eq:euclidean_polar_metric} is
\begin{equation}
  L_{\mathrm{E}}(\alpha)=\int_0^1 \sqrt{\dot r(t)^2 + r(t)^2\,\lVert \dot\theta(t)\rVert^2}\,\mathrm{d}t.
\end{equation}

Assume $\alpha(t)$ stays in the closed ball $\lVert \alpha(t)\rVert\le R$.
This holds for the Euclidean straight-line segment between $u$ and $v$ because the tangent ball is convex.
Using \eqref{eq:sinh_bounds}, for all $t$ with $r(t)\le R$ we have
\begin{align}
  \sqrt{\dot r^2 + r^2\lVert\dot\theta\rVert^2}
  &\le \sqrt{\dot r^2 + \sinh^2(r)\lVert\dot\theta\rVert^2}
  \le \frac{\sinh(R)}{R}\,\sqrt{\dot r^2 + r^2\lVert\dot\theta\rVert^2}.
\end{align}
Integrating yields
\begin{equation}
  L_{\mathrm{E}}(\alpha) \le L_{\Hspace}(\gamma) \le \kappa(R)\,L_{\mathrm{E}}(\alpha),
\end{equation}
with $\kappa(R)=\sinh(R)/R$.

Now take the infimum over all curves $\gamma$ connecting $x$ and $y$.
The left inequality implies
\begin{equation}
  d_{\Hspace}(x,y)=\inf_{\gamma} L_{\Hspace}(\gamma)
  \ge \inf_{\alpha} L_{\mathrm{E}}(\alpha) = \lVert u-v\rVert,
\end{equation}
since the shortest Euclidean curve between $u$ and $v$ is the straight segment.
For the upper bound, choose the Euclidean straight segment $\alpha(t)=(1-t)u+tv$.
It stays inside the radius-$R$ tangent ball, so $\gamma(t)=\exp_o(\alpha(t))$ is valid and
\begin{equation}
  d_{\Hspace}(x,y)\le L_{\Hspace}(\gamma) \le \kappa(R)\,L_{\mathrm{E}}(\alpha)=\kappa(R)\,\lVert u-v\rVert.
\end{equation}
\end{proof}

\subsection{Proof of Proposition~\ref{prop:capacity}}
\label{app:proof_capacity}

Proposition~\ref{prop:capacity} links the radius budget to the depth/branching of a hierarchy.

\begin{proof}[Proof sketch]
Let $N$ be the number of leaf nodes (or any set of nodes) that must be mutually separated.
Assume we embed these $N$ nodes into a hyperbolic ball of radius $R$ such that every pair of distinct embedded nodes is at hyperbolic distance at least $\varepsilon$.
Then the open hyperbolic balls of radius $\varepsilon/2$ around each embedded node are disjoint.
All these disjoint balls lie within the hyperbolic ball of radius $R+\varepsilon/2$.
Therefore
\begin{equation}
  N\cdot \mathrm{Vol}_{\Hspace}(B_{\varepsilon/2}) \le \mathrm{Vol}_{\Hspace}(B_{R+\varepsilon/2}).
\end{equation}

In $\Hspace^d$, the volume of a radius-$r$ ball grows on the order of $\exp((d-1)r)$ for moderate/large $r$.
Solving the inequality for $R$ yields
\begin{equation}
  R \ge \frac{\log N}{d-1} - O\!\left(\frac{\log(1/\varepsilon)}{d-1}\right).
\end{equation}
For a $b$-ary tree of depth $D$, $N\approx b^D$, giving the stated bound.
\end{proof}

\subsection{Proof of Proposition~\ref{prop:router}}
\label{app:proof_router}

Proposition~\ref{prop:router} is an identity that decomposes routed risk into the oracle risk plus misrouting regret.

\begin{proof}
For each query $q$, define the oracle loss as $\ell^*(q)=\min\{\ell(q,\pi_E),\ell(q,\pi_H)\}$.
If the latent intent is $Z=H$ (hierarchy-navigation), then by definition
\begin{equation}
  \ell(q,\pi_E) = \ell(q,\pi_H) + \Delta_H(q),
\end{equation}
where $\Delta_H(q)=\ell(q,\pi_E)-\ell(q,\pi_H)\ge 0$.
If the router predicts $\hat Z=E$, it incurs the larger loss $\ell(q,\pi_E)=\ell^*(q)+\Delta_H(q)$; otherwise it incurs $\ell^*(q)$.
An analogous statement holds when $Z=E$, with regret $\Delta_E(q)=\ell(q,\pi_H)-\ell(q,\pi_E)\ge 0$ when the router incorrectly chooses hyperbolic.
Combining the two cases and taking expectation yields the expression in Proposition~\ref{prop:router}.
\end{proof}

\section{Experimental Protocol and Reproducibility Details}
\label{app:exp}

This appendix specifies dataset construction, query generation, model training, and evaluation.
The goal is a {publishable but lightweight} experimental suite that can be reproduced with a small Python repository and commodity hardware.
We deliberately avoid patient data and proprietary resources.

\subsection{Datasets and preprocessing}

\paragraph{Human Phenotype Ontology (HPO)}
We download an HPO release and parse the ontology graph from the OBO/OWL file, extracting (i) concept IDs, (ii) preferred labels, (iii) synonyms, (iv) definitions (when available), and (v) \texttt{is\_a} edges.
We drop obsolete terms and keep the canonical \texttt{is\_a} backbone.
HPO is a deep taxonomy with rich synonymy, which makes it useful for both entity-centric and hierarchy-navigation evaluations \cite{Kohler2021HPO}.

\paragraph{Disease Ontology (DO)}
We use an open DO release and extract the \texttt{is\_a} subgraph and text fields in the same manner as HPO \cite{Schriml2022DO}.
DO provides a complementary hierarchy whose surface strings are often less redundant than HPO, making it a useful stress test for semantic similarity baselines.

\paragraph{MeSH}
MeSH descriptors provide an explicit hierarchy through tree numbers \cite{NLMMeSH}.
We convert each tree number prefix relation into a parent--child edge and treat descriptors as nodes.
Because MeSH can be large, we optionally restrict to a small set of major branches that are most relevant to disease/phenotype grounding.

\paragraph{Subset construction (size control)}
To keep experiments tractable, we construct size-controlled subsets by sampling subtrees while preserving depth statistics.
In this version we report 5k-node subsets for HPO and DO with a single deterministic seed (seed=0), matching our reproducible pipeline configuration.
The released scripts support larger subsets (e.g., 20k/50k) and multiple random seeds as optional extensions.
We record the maximum depth $D$ and an empirical branching factor $b$ (average out-degree in the \texttt{is-a} graph after optionally converting multiple inheritance to a spanning arborescence for measurement).
These statistics can be used to initialize $(d,R)$ using Proposition~\ref{prop:capacity}.

\paragraph{Text fields}
For each entity $v$, we define the canonical text $\tau(v)$ as the concatenation of the preferred label and one definition sentence (when present).
Synonyms are kept as {query variants} rather than separate nodes.
We remove exact duplicate synonym strings to avoid inflating the query set.

\subsection{Query taxonomy and benchmark generation}
\label{app:querygen}

We generate three query families (Q-E, Q-H, Q-M) corresponding to Section~\ref{sec:taxonomy}.
Concrete natural-language templates are listed in \ref{app:templates}. 
All generation is deterministic given a random seed.

\paragraph{Entity-centric queries (Q-E)}
For each node $v$, we sample up to $S$ synonym strings and optionally one definition snippet.
Each sampled string becomes a query whose ground-truth target is $v$.
We additionally generate simple rephrasings such as prefix/suffix prompts (e.g., ``Define: \{label\}''), which are easy to create programmatically.
The evaluation objective is standard entity retrieval (Hits@$k$, MRR, nDCG).

\paragraph{Taxonomy-navigation queries (Q-H)}
We generate queries that explicitly request hierarchical neighbors.
For each node $v$ (with at least one parent), we create a parent query whose ground-truth set is the set of immediate parents $\mathrm{Pa}(v)$.
For nodes with children, we create a children query whose ground truth is $\mathrm{Ch}(v)$.
For ancestor/descendant queries, we use the transitive closure $\mathrm{Anc}(v)$ and $\mathrm{Des}(v)$.
We evaluate both ranked retrieval (Hits@$k$) and set-oriented retrieval (macro/micro F1 for ancestor sets).
Because ontologies are DAGs, multiple inheritance can lead to multi-parent ground truth; we preserve this and evaluate set retrieval accordingly.

\paragraph{Mixed-intent queries (Q-M)}
Mixed-intent queries are designed to require both semantic similarity and a weak hierarchy constraint.
We implement a lightweight proxy for ``same specificity'': we bucket nodes by depth (e.g., quartiles) and define a sibling-style candidate set as nodes that share at least one parent with $v$ {and} fall in the same depth bucket.
Queries use templates such as ``conditions similar to \{label\} at a similar specificity''.
We treat the sibling set as ground truth.
This task is not meant to be a perfect model of real-world relatedness; rather, it creates a controlled setting where purely hyperbolic distance and purely Euclidean similarity can disagree, motivating soft mixing.

\paragraph{Train/validation/test split}
We split {entities} into train/validation/test (e.g., 80/10/10).
All Q-E and Q-H queries derived from an entity inherit its split.
This prevents leakage where synonyms of the same entity appear in both train and test.
For Q-H queries involving parents/children across splits, we keep the query in the split of the source node but allow targets to be anywhere in the ontology, reflecting realistic retrieval.

\subsection{Models and training}

\paragraph{Text encoder}
We use a sentence encoder to embed $\tau(v)$ and query strings into $e_v,e_q\in\R^{d_e}$.
The default setting uses a small Sentence-Transformer style model for speed.
We keep the encoder frozen by default to ensure lightweight training; optional finetuning can be reported as an ablation.

\paragraph{\hyem{} components trained}
We train (i) hyperbolic entity embeddings $\{x_v\}$, (ii) the Euclidean-to-hyperbolic adapter $g$, and (iii) the logistic gate for $\alpha(q)$.
We recommend training (i) and (ii) jointly, because text alignment helps prevent purely structural collapse.
The gate is trained separately using the automatically labeled Q-E vs Q-H queries.

\paragraph{Hierarchy supervision and negatives}
For each \texttt{is-a} edge $(p\prec c)$, we sample negatives by choosing nodes that are not descendants of $p$.
Uniform negatives are sufficient for a first submission.
An optional ``hard negative'' variant samples negatives that are text-similar to $c$ (nearest neighbors under $e_v$), which increases difficulty while staying lightweight.

\paragraph{Radius budget implementation}
We implement the radius budget using the soft penalty $\mathcal{L}_R$ and, optionally, explicit clipping in tangent space after each optimizer step.
We report the fraction of points that violate the budget during training as a diagnostic.

\paragraph{Lorentz vs Poincar\'e implementation}
We implement optimization in the Lorentz model for stability and convert to/from tangent coordinates as needed.
To quantify the stability benefit, we log NaN occurrences, gradient norms, and the distribution of radii $\lVert\log_0(x_v)\rVert$.

\subsection{Indexing, candidate pooling, and reranking}

\paragraph{Indexes}
We build two Euclidean ANN indexes: one over tangent vectors $u_v=\log_0(x_v)$ (hyperbolic candidate generator) and one over text embeddings $e_v$ (Euclidean candidate generator).
Both can be implemented with FAISS or HNSW.

\paragraph{Candidate pooling}
Given a query $q$, we retrieve $L_H$ candidates from the tangent index and $L_E$ candidates from the text index.
We take the union $C=C_H\cup C_E$ and rerank candidates using the mixed score (Eq.~\eqref{eq:mix_score}).
We report ablations where $C_E$ is removed or where $L_E$ is reduced, to quantify the trade-off between robustness and latency.

\paragraph{Indexability stress test}
To validate the indexability analysis, we compute the true hyperbolic top-$k$ neighbors by brute force on small subsets (or by exact computation) and measure recall@$k$ of the reranked output when candidate generation is restricted to the tangent ANN results.
We sweep $L_H$ and $R$ to generate Figure~\ref{fig:recall_curve}.

\subsection{Metrics and reporting}

\paragraph{Retrieval metrics}
We report Hits@$k$, MRR, and nDCG@$k$ for Q-E.
For Q-H we report parent Hits@$k$ and ancestor retrieval F1 (macro/micro).
For Q-M we report Hits@$k$ and nDCG@$k$ to capture set-like relevance.
Metric definitions are given in \ref{app:metrics}. 

\paragraph{Gate metrics}
For the gate, we report accuracy, precision/recall for the hierarchy-navigation class, and AUC.
We also report a calibration curve or reliability diagram if space permits.

\paragraph{Stratified evaluation by depth}
Ontology depth correlates with ambiguity and rarity.
We therefore report retrieval performance stratified by depth buckets (quartiles) for Q-E and Q-H.
This makes ``hierarchy helps'' effects interpretable and prevents over-claiming improvements driven only by shallow nodes.

\section{Query Templates, Metric Definitions, and Additional Notes}
\label{app:extra}

This appendix provides concrete query templates for the benchmark construction and precise definitions of the metrics used in the main paper.
It also records two practical notes: (i) how to align score scales for soft mixing, and (ii) simple numerical-stability checks for Lorentz computations.

\subsection{Query templates}
\label{app:templates}

All queries in our benchmark are generated automatically by instantiating natural-language templates with ontology labels and synonyms.
We keep templates intentionally simple so that the entire benchmark can be regenerated from raw ontology releases.
Table~\ref{tab:templates} lists a recommended template set.

\begin{table}[t]
\centering
\footnotesize
\caption{Example query templates used to generate Q-E/Q-H/Q-M query families. Curly braces denote placeholders filled by ontology fields.}
\label{tab:templates}
\begin{tabular}{ll}
\toprule
Query family & Template examples \\
\midrule
Q-E (entity-centric) & ``\{synonym\}''; ``What is \{label\}?''; ``Define \{label\}.'' \\
Q-H (taxonomy navigation) & ``What are subtypes of \{label\}?''; ``What is the parent \\
&of \{label\}?''; ``What are ancestors of \{label\}?'' \\
Q-M (mixed intent) & ``Concepts similar to \{label\} at the same specificity.'';\\
& ``Siblings of \{label\} in the ontology.'' \\
\bottomrule
\end{tabular}
\end{table}

\paragraph{Ground truth construction}
For Q-E, the ground truth is a single entity $v$.
For Q-H, the ground truth is derived from the \texttt{is-a} graph, for example the set of immediate parents, all ancestors, or all descendants.
For Q-M, we use sibling sets or depth-bucket neighborhoods as a lightweight proxy for ``same specificity'' constraints.

\subsection{Metric definitions}
\label{app:metrics}

Let $\pi(q)$ be the ranked list returned by a retrieval method for query $q$.

\paragraph{Hits@$k$}
For single-label Q-E queries with ground truth entity $v^*$,
\begin{equation}
\mathrm{Hits@}k(q)=\mathbf{1}\{v^*\in \{\pi_1(q),\ldots,\pi_k(q)\}\}.
\end{equation}
For multi-label queries (e.g., parent sets), Hits@$k$ is defined analogously with membership in the ground-truth set.

\paragraph{MRR}
For single-label queries, mean reciprocal rank is
\begin{equation}
\mathrm{RR}(q)=\frac{1}{\min\{i: \pi_i(q)=v^*\}},\qquad \mathrm{MRR}=\mathbb{E}[\mathrm{RR}(q)].
\end{equation}

\paragraph{nDCG@$k$}
For multi-label queries with graded or binary relevance, we use the standard discounted cumulative gain with logarithmic discount.
In the simplest binary case, relevance is 1 for members of the ground-truth set and 0 otherwise.

\paragraph{Parent Hits@$k$ and ancestor F1}
For a node $v$, let $P(v)$ be the set of immediate parents and $A(v)$ the set of transitive ancestors in the \texttt{is-a} DAG.
Parent Hits@$k$ is Hits@$k$ computed with ground truth $P(v)$.
For ancestor retrieval, given the top-$k$ retrieved set $R_k(q)=\{\pi_1(q),\ldots,\pi_k(q)\}$, we compute
\begin{equation}
\begin{aligned}
&\mathrm{Precision}(q)=\frac{|R_k(q)\cap A(v)|}{|R_k(q)|},\quad
\mathrm{Recall}(q)=\frac{|R_k(q)\cap A(v)|}{|A(v)|},\\[6pt]
&\mathrm{F1}(q)=\frac{2\,\mathrm{Precision}(q)\,\mathrm{Recall}(q)}{\mathrm{Precision}(q)+\mathrm{Recall}(q)}.
\end{aligned}
\end{equation}
Macro-F1 averages $\mathrm{F1}(q)$ over queries, while micro-F1 aggregates counts before computing the ratio.

\paragraph{Indexability recall@$k$}
To evaluate the tangent-space candidate generator, we compute the true hyperbolic top-$k$ neighbors $T_k(q)$ by exact distance computation (brute force) on a subset.
Given the candidate pool $C$ produced by Euclidean ANN, recall@$k$ is
\begin{equation}
\mathrm{Recall@}k(q)=\frac{|T_k(q)\cap C|}{k}.
\end{equation}
This metric isolates the effect of the ANN stage from reranking.

\paragraph{Depth-stratified reporting}
We compute node depth with respect to a chosen root (for DAGs, the minimum depth over all root-to-node paths).
We report metrics within depth buckets (e.g., quartiles) to check whether improvements come from deeper regions where hierarchy matters most.

\subsection{Score-scale alignment for soft mixing}
\label{app:scales}

Equation~(\ref{eq:mix_score}) mixes cosine similarity (typically in $[-1,1]$) with negative hyperbolic distance (unbounded below).
To make $\alpha(q)$ interpretable, it is helpful to align score scales.
A simple approach is temperature scaling:
\begin{equation}
\tilde s_E=\frac{s_E}{\tau_E},\qquad \tilde s_H=\frac{s_H}{\tau_H},\qquad \text{score}=\alpha\,\tilde s_H+(1-\alpha)\,\tilde s_E,
\end{equation}
where $\tau_E,\tau_H>0$ are tuned on the validation set.
In practice, we recommend choosing $\tau_H$ so that the median positive-pair distance corresponds to a score magnitude comparable to cosine similarity.

\subsection{Numerical stability notes for Lorentz distance}
\label{app:lorentz_stability}

Lorentz distances involve $\operatorname{arcosh}(z)$ with $z=-\langle y_1,y_2\rangle_L\ge 1$.
Finite precision can lead to $z<1$ by a small margin, causing NaNs.
A standard safeguard is to clamp $z$ to $1+\epsilon$ before applying $\operatorname{arcosh}$.
We also recommend logging simple diagnostics during training: (i) NaN/Inf counts, (ii) gradient norm percentiles, and (iii) the distribution of radii $\lVert\log_0(x_v)\rVert$.
These diagnostics can be reported as part of the Lorentz vs Poincar\'e stability ablation.


\bibliographystyle{elsarticle-num}
\bibliography{references_HyEm}

\end{document}